\documentclass[draftclsnofoot,onecolumn]{IEEEtran}
\usepackage{amsmath,amsfonts}
\usepackage{algorithmic}
\usepackage{algorithm}
\usepackage{array}
\usepackage[caption=false,font=normalsize,labelfont=sf,textfont=sf]{subfig}
\usepackage{textcomp}
\usepackage{stfloats}
\usepackage{url}
\usepackage{verbatim}
\usepackage{graphicx}
\usepackage{cite}

\usepackage{mathrsfs}
\usepackage{amssymb}
\usepackage{amsthm}
\usepackage{color}
\usepackage{hyperref}
\usepackage{bm}

\renewcommand{\textcolor}[3][]{#3}
\renewcommand{\textcolor}[3][]{%
	\ifx\relax#1\relax#3\else
	\ifx#1blue#3\else\textcolor[#1]{#2}{#3}\fi
	\fi}
\def\BibTeX{{\rm B\kern-.05em{\sc i\kern-.025em b}\kern-.08em
    T\kern-.1667em\lower.7ex\hbox{E}\kern-.125emX}}
\usepackage{balance}
\newcommand{\bma}{\bm{a}}
\newcommand{\x}{\bm{x}}
\newcommand{\y}{\bm{y}}
\newcommand{\e}{\bm{e}}
\newcommand{\g}{\bm{g}}
\newcommand{\z}{\bm{z}}

\newcommand{\h}{\bm{h}}

\newcommand{\T}{\mathcal{T}}
\newcommand{\K}{\mathcal{K}}

\newcommand{\C}{\mathcal{C}}
\newcommand{\M}{\mathcal{M}}

\DeclareMathOperator{\rad}{rad}
\DeclareMathOperator{\diam}{diam}

\DeclareMathOperator{\mismatch}{mismatch}
\DeclareMathOperator{\conv}{conv}

\newtheorem{theorem}{Theorem}[section]
\newtheorem{definition}{Definition}[section]
\newtheorem{lemma}{Lemma}[section]
\newtheorem{corollary}{Corollary}[section]
\newtheorem{remark}{Remark}[section]

\begin{document}

\title{Stability of Constrained Optimization Models for
Structured Signal Recovery}

\author{
	Yijun Zhong,
	~Yi Shen
\thanks{This work is supported by
the NSFC under  grant No. 12371101 and 11901529. (\textit{Corresponding author: Yi Shen})}
\thanks{Yijun Zhong is with the Department of Mathematics, Zhejiang Sci-Tech University, Hangzhou,  China, 310018. (E-mail: zhongyijun@zstu.edu.cn).}
\thanks{Yi Shen is with the Department of Mathematics, Zhejiang Sci-Tech University, Hangzhou,  China, 310018. (E-mail: yshen@zstu.edu.cn).
}
}

%
\markboth{Journal of \LaTeX\ Class Files}%
{Stability of Constraint Optimization Models for Structured Signal Recovery}
\maketitle

\begin{abstract}
Recovering an unknown but structured signal from its measurements is a challenging problem with significant applications in fields such as imaging restoration, wireless communications, and signal processing. 
	In this paper, we consider 
	the inherent problem stems from the prior knowledge about the signal’s structure—such as sparsity which is critical for signal recovery models.
We investigate three constrained optimization models that effectively address this challenge, each leveraging distinct forms of structural priors to regularize the solution space. Our theoretical analysis demonstrates that these models exhibit robustness to noise while maintaining stability with respect to tuning parameters that is a crucial property for practical applications, when the parameter selection is often nontrivial. 
By providing  theoretical foundations, our work supports their practical use in scenarios where measurement imperfections and model uncertainties are unavoidable.
Furthermore, under mild conditions, we establish trade-off
between the sample complexity and the mismatch error.

\end{abstract}


\section{Introduction}
\IEEEPARstart{M}{odern} information processing and machine learning continually confront the challenge of effectively handling structured signals, whether they are data, images, or sound.
In this paper, we primarily focus on the recovery of  vectors $\x^*\in \mathbb{R}^n$ from 
a relatively small number of
noisy measurements $\y \in \mathbb{R}^m$ and the measurement matrix $A\in \mathbb{R}^{m\times n}$ with $m<n$. Since this problem is generally ill-posed, one can recover $\x^*$ when it possesses some type of structure. We assume that prior information about the bounded structure-inducing function $f$ at $\x^*$ is available, then one can define a feasible set
\begin{displaymath}
	\K = \{\x\in\mathbb{R}^n:f(\x)\leq \eta \},
\end{displaymath}
where $\eta =f(\x^*)$ is called the \textit{optimal tuning parameter}. 
The feasible set  $\mathcal{K}$  
which captures the structure of $\x^*$
encompasses many different types, such as the set
of all $s$-sparse signals, 
an appropriately scaled $l_1$ ball
and a star-shaped set \cite{Plan2017}. 
We consider  
three different optimization models for structure signal recovery with the feasible set $\mathcal{K}$ as follows:
\begin{itemize}
	\item Constrained least squares
	\begin{equation}\label{LS}
		\min_{\x\in \mathcal{K}}\frac{1}{2}\|\y-   A\x\|_2^2.
	\end{equation}
	\item Constrained least absolute deviation 
	\begin{equation}\label{LAD}
		\min_{\x\in \mathcal{K}}  \left\|\y-A\x\right\|_1.
	\end{equation}
	\item Constrained nonlinear least squares
	\begin{equation}\label{non-LS}
		\min_{\x\in \mathcal{K}} \frac{1}{2}
		\left\|\y-|A\x|\right\|_2^2.
	\end{equation}
\end{itemize}

\subsection{Motivation} 
The feasible set $\mathcal{K}$ offers a defined range of possible solutions, allowing for a more focused and efficient approach to problem-solving in various applications. 
However, determining the optimal choice of the tuning parameter 
$\eta$
is not always straightforward. A widely accepted method for approximating a suitable value is through cross-validation. It is natural to question whether the performance of constrained optimization models is influenced by the specific selection of the tuning parameter. Research on the sensitivity of sparse signal recovery models to parameter selection can be found in
\cite{Chatterjee2014,ver2016,Berk2021,Berk2022,Berk2023,karma2019} and references therein.
In particular, the studies in \cite{Berk2021,Berk2022}  evaluated the minimax order-optimal recovery results of three LASSO variants concerning their governing parameters. 
In \cite{Berk2021}, the authors investigated the proximal denoising problem, characterizing the asymptotic singularity of the risk as the noise scale tends to zero. 
Specifically  \cite[Theorem 2.1]{Berk2021} considered
three distinct regimes:  $\|\x^*\|_1<\eta$,  $ \|\x^*\|_1=\eta$, and  $ \|\x^*\|_1>\eta$. Then \cite[Theorem V.1 (Asymptotic Singularity)]{Berk2022} extended 
this analysis to
the constrained Lasso 
 with $\|\x^*\|_1 \neq \eta$ .
Motivated by these asymptotic results on the tuning parameter $\eta$, we   develop  non-asymptotic error bounds for structure signal recovery  problem when $ f(\x^*) \neq \eta$.	

\subsection{Contributions}
We establish the stability  of the \textcolor{blue}{constrained} optimization model \textcolor{blue}{\eqref{LS},  \eqref{LAD}
and \eqref{non-LS}
}  in the following two cases
\begin{equation*}
 f(\x^*) < \eta 
\quad 
\text{and}
\quad 
 f(\x^*) \ge \eta.
\end{equation*}
The condition $f(\x^*)\neq \eta$ can be understood from the intuition: the approximation error is controlled by the effective dimension of the constraint set \cite{Berk2021}.
One key finding is the stability of both linear and non-linear estimators against tuning parameter inaccuracies, demonstrating their robustness even when optimal hyperparameters are not precisely known. 
This stability is particularly crucial in practical applications where exact parameter tuning may be challenging due to computational constraints or limited prior knowledge.
On the other hand, our main results quantify the trade-off between sample complexity (number of measurements) and the mismatch error (distance between the  $f(\x^*)$ and the tuning parameter $\eta$ of the feasible set $\K$). 
This refined trade-off provides a  theoretical understanding of the interplay between measurement constraints and tuning parameter, offering practical guidelines for algorithm design in compressed sensing, high-dimensional statistics, and inverse problems.

\subsection{Road Map}

The outline of this paper is structured as follows. In the remainder of this section,we describe three typical models concerning linear and nonlinear measurements.
Section \ref{related} reviews the relevant work on stability analysis for the \eqref{LS}, \eqref{LAD} and
\eqref{non-LS}, along with algorithms for solving these problems.
Section \ref{sec:gaussian}  briefly recalls the Gaussian width and its applications to feasible sets. Section \ref{sec:main} establishes the stability of the three models discussed in this paper. Some proofs of the Lemmas and Theorems are provided in the supplemental material.

\subsection{Notation}
For any given vector 
$\x= (x_1,x_2,\ldots,x_n)^T$,
the  ``$\ell_0$--norm'' of $\x$  denoted by $\|\x\|_0$,
is the count of its nonzero entries.
A vector is said to be $s$-sparse if $\|\x\|_0\le s$. 
For any $p>0$, we define
$\|\x\|_p:=
\left(\sum_{i=1}^{n}|x_i|^p\right)^{1/p}$.
For any given positive integer $m$, we denote 
$[m]=\{1,\ldots,m\}$.
Given an index set $\Omega \subset [m]$ and a  vector $\x$, let
$\x_{\Omega}$ denote the vector whose
$i$-th entry is equal to $i$-th entry of $\x$  for $i$  in $\Omega$ and  equal to zero otherwise.
Similarly,
let $A_{\Omega}$ denote the matrix  whose $i$-th row is equal to $i$-th row of $A$ for $i$ in $\Omega$ and equal to zero vector otherwise.
The letters $C$, $C_1$, and $C_2$ are usually treated as  constants, but their value vary across different parts of the paper. Instead of explicitly writing $a \leq C   b$, we write $a \lesssim b$, and instead of $a \geq C   b$, we write $a \gtrsim b$, with $C > 1$, respectively. We use the notation $a \approx b$ to indicate that there exist constants $C_1 > 0$ and $C_2 > 0$ such that $C_1 b \leq a \leq C_2 b.$

Let $\mathcal{N}(0,I_n)$ stand for the multivariate normal distribution in $\mathbb{R}^n$ with
zero mean and   covariance matrix  identity $I_n$. 
We assume that the measurement matrix $A\in \mathbb{R}^{m\times n}$ is a Gaussian random matrix  whose rows $\bma_i\sim \mathcal{N}(0,I_n)$ are independent.
The unit Euclidean
sphere is denoted by
\begin{equation*}
\mathcal{S}^{n-1}=\{\x\in\mathbb{R}^n:\ \|\x\|_2=1\},
\end{equation*}
and the unit Euclidean ball in $\mathbb{R}^n$ is denoted by 
\begin{equation*}
\mathcal{B}_2^n=\{\x\in\mathbb{R}^n:\ \|\x\|_2\leq 1\}.
\end{equation*}
The Euclidean projection of $\x^*$ onto the set $\mathcal{K}$ is denoted by 
\begin{equation*}
\bm{P}_{\mathcal{K}}(\x^*)
\in \arg\min_{\x\in \mathcal{K}}
\| \x^* - \x  \|_2^2.
\end{equation*}
Notice that we do not require $\mathcal{K}$ be convex, thus the projection operator $\bm{P}_{\mathcal{K}}$ may not be unique.
Denote 
$d(\x,\y)=(\mathbb{E}(X_{\x}-X_{\y})^2)^{\frac{1}{2}}$ as the canonical metric on the index set $\T$ for a given random process $(X_{\x})_{\x\in\T}$.
For any given $\T$,
we denote the diameter of set $\T$ by
\begin{equation*}
\diam(\T)=\sup\{\|\x-\y\|_2:\x,\y\in\T\}
\end{equation*}
and
 the radius of $\T$ by ${ \rad}(\T)=\sup\limits_{\x\in \T}\|\x\|_2$.
We also use the notion of Minkowski functional of $\K$ which is defined as
\[
\|\x\|_{\K}=\inf\{\lambda>0:
\ \lambda^{-1}\x\in\K\}
\]
for the given vector $\x\in \mathbb{R}^n$.

\section{Related Works}\label{related}

Within the past few years, there have been numerous studies in the literature on the stability of signal recovery models with structural properties, including the convergence analysis of algorithms for solving such models.
These works draw significant inspiration from the relatively recent field of compressed sensing \cite{foucart2013invitation}, its theoretical foundations lie in classical results from geometric functional analysis \cite{Gordon1988} and convex integral geometry \cite{schneider2008}. 
The related works on sparse recovery problems can be broadly classified into two categories:
the linear and non-linear approaches. For conciseness, we focus on the real-valued scenario.

\subsection{Linear Estimation}
We begin with the linear inverse problems:
\begin{equation*}
	\y= A\x^*+\e
\end{equation*}
where $\bm{e}$ denotes noise vector that is independent from the measurement matrix.
The first  recovery model is the standard least square method subject to a structural \textcolor{blue}{constrained} \eqref{LS}.
Different structural information of the signal $\x$
represented by 
$\mathcal{K}$ gives rise to numerous extensively studied problems.
For instance, if $f$ is the sparsity-induced 
``$l_0$ norm'', then model \eqref{LS} is the Subset Selection. 
If $f$ is the 
$l_1$ norm, then model \eqref{LS} is called the 
\textit{Basis Pursuit} (BP)
in compressed sensing literature \cite{chen1998} or the
\textit{Least Absolute Shrinkage and Selection Operator }
(Lasso) in 
in the context of statistical regression \cite{Tibshirani1996}.
In \cite{Chatterjee2014}, Chatterjee considered the least squares under a convex constraint 
which is
a special case of \eqref{LS}
and illustrated that for nonsingular design matrices, the prediction error is vastly small when $f(\x^*)=\eta$. 
A vast body of work has studied algorithms for solving sparsity-constrained optimization \eqref{LS}, 
see e.g.,
\cite{Bruckstein2009,Bahmani2013,yuan2014,oymak2018,Mahdi2017,BLUMENSATH2009265,Bahmani2013A} and  many references therein.  
If $f$ is the non-convex sparsity-induced $l_p$ norm with $0< p\le 1$, then
two kinds of iterative algorithms have been used to approximate $\x^*$ by solving the $l_p$-constrained least square:
\begin{equation}\label{lp}
	\min\limits_{\|\x\|_p\leq  \eta } \frac{1}{2}\|\y-A\x\|^2_2.
\end{equation}
One is the project gradient descent algorithm
\cite{Bahmani2013A,oymak2018}. 
In the case that $ \|\x^*\|_p \ge \eta  $, the project gradient descent algorithm is proved to be stable
to  tuning parameters $\eta$ 
under the \textit{Restricted Isomerty Property} (RIP) conditions \cite{Bahmani2013A}. 
Then
the project gradient descent algorithm with a fixed learning rate is proved in \cite{oymak2018} to be stable to  tuning parameters in both case 
$  \|\x^*\|_p \le  \eta  $ 
and the case 
$\|\x^*\|_p \ge \eta  $.
The stability analysis on the iterative algorithms in \cite{Bahmani2013A,oymak2018,zhong2024} are consistent with Theorem \ref{th1}.

For linear models with heavy-tailed noise or outliers, least squares \eqref{LS} often perform poorly for signal recovery, necessitating the use of the  least absolute deviation (LAD) model \eqref{LAD}. 
If $f$ is chosen to the $l_1$ norm and the tuning
parameter 
satisfies 
$ \|\x^*\|_1 \le  \eta  $,
then 
the model \eqref{LAD} 
\begin{equation}\label{LAD2}
	\min_{\|\x\|_1 \le \eta} 
	\left\|\y-A\x\right\|_1
\end{equation}
is showed 
to be highly robust for both
dense noise and adversarial sparse noise 
\cite{karma2019}.
However,  the case of \textcolor{blue}{$\|\x^*\|_1  > \eta$}  is not addressed in
\cite{karma2019}. 
\textcolor{blue}
{We shall supplement   \cite[Theorem 1.1]{karma2019} in the discussion after Corollary \ref{l1-lad}.}
Reducing the
adversarial sparse noise 
by solving  
the \textcolor{blue}{constrained} LAD
was further extended to  low rank matrix recovery problems \cite{Xu2022} and phase 
retrieval  problems \cite{huang2023}. 
If $f$ is chosen to the ``$l_0$--norm'' in the model \eqref{LAD}, then
projected subgradient descent methods were proposed in \cite{li2023,liu2019} to solve \eqref{LAD} directly.
A  well-known model in statistics is the $l_1$ penalized LAD. The asymptotic
properties of variable selection consistency were discussed \cite{gao2010}. 
The 
consistency of the $l_1$ penalized LAD
estimator was discussed in \cite{Sophie2011, wang2007}.
The near oracle performance of the $l_1$ penalized LAD is obtained in \cite{WANG2013135}.
\subsection{Nonlinear Estimation--Phase Retrieval}
Our analysis method for linear models can be extended to other related recovery problems, such as phase retrieval, low-rank matrix recovery, and blind deconvolution. We illustrated this by confirming the stability of a well-known phase retrieval model.
Mathematically, measurements formulated by the phase retrieval model are 
\begin{equation*}
	\y=|A\x^*|+\bm{e},
\end{equation*}
where $\bm{e}$ denotes the noise  that is independent from the measurement matrix.
One approach to 
recovery $\x^*$ is solving
the model \eqref{non-LS}. This quadratic measurement scheme, despite its nonlinear nature, has proven effective in capturing intensity-only observations while maintaining mathematical tractability.
The model \eqref{non-LS}
works rather well in some
practical domains such as ptychography for chip imaging \cite{Mahdi}. 
\textcolor{blue}{Compared to the intensity based model, algorithms minimizing the amplitude based model \eqref{non-LS} are usually more efficient in computation \cite{wang2017,zhang2017,Mahdi}.}
The  model \eqref{non-LS} was proved to be robust to the noise 
by Huang and Xu 
in \cite{xu2020} 
\textcolor{blue}{in the case where the tuning parameter is optimal, i.e.  
$\|\x^*\|_1  = \eta$ .
Moreover, the result in \cite[Theorem 1.2]{xia2024-1} extends the error bound in \cite[Theorem 1.5]{xu2020} for the real signals to complex
ones by employing different tools.
 Two algorithms for solving the phase retrieval problem with generative priors are proposed in \cite{hyder2019}.
}
The
\textit{Projected Wirtinger Flow} (PWF) is proposed to
solve the model \eqref{non-LS}
in  \cite{Mahdi}.  
The convergence guarantees are provided as long as
PWF is initialized  in a proper neighborhood of the unknown
signal $\x^*$.
\textcolor{blue}{In order to solve the model \eqref{non-LS} with the optimal condition
$\|\x\|_0 \le \eta$ with
$\|\x^*\|_0 = \eta  $,
many algorithms have been designed and achieved good results on the problem of sparse phase retrieval, 
such as Iterative Hard Thresholding \cite{cai2022}, Truncated Amplitude Flow \cite{wang2017}, CoPRAM \cite{Jagatap2019}.
The existing works [30-34] focus on  stability analysis of models or  convergence analysis of algorithms, all of which assume that 
$\eta$
is the optimal parameter. 
This oversight means that the potential impact of 
$\eta$ deviation from optimality on model stability and algorithm convergence remains unexplored, leaving a gap in our understanding of these analytical frameworks.
}

\section{Preliminaries}\label{sec:gaussian}
\textcolor{blue}{
In this section, we recall some properties of the Gaussian width, which is extensively used in the context of convex recovery of structured signals from independent random linear measurements \cite{Tropp2015,liu2020sparsity}, and
 in statistics and signal processing to quantify the variability or dispersion of a set of data points \cite{oymak2018,Vershynin2020,Ver2015}.}
\begin{definition}[Gaussian width]
The Gaussian width of a set 
$\K\subset \mathbb{R}^n$
is defined as:
\[
\omega(\K)=\mathbb{E}_{\g}
\left[
\sup\limits_{\z\in\K}\langle \g,\z\rangle
\right],
\]
where the expectation is taken over  standard Gaussian random vectors $\g\sim \mathcal{N}(0,I_n)$.
\end{definition}
The Gaussian width is similar to the mean width which originates from geometric functional analysis and convex geometry. 
For a given set $\K$, 
the Gaussian width $\omega(\K)$ can be bounded by the covering number of $\K$ \cite[Theorem 3.11]{Ver2015}.
\begin{definition}[Descent set and Cone] \cite[Definition 1]{oymak2018}
The set of descent of the function $f$ at a point $\x$ is defined as
\[
\mathcal{D}_f(\x)=\{\h:~f(\x+\h)\leq f(\x)\}.
\]
The cone of descent is defined as a \textcolor{blue}{closed} cone $\mathcal{C}_f(\x)$ that contains the descent set, i.e., $\mathcal{D}_f(\x)\subset \mathcal{C}_f(\x)$. The tangent cone is the conic hull of the descent set.
\end{definition}
The alternative definition of $\mathcal{C}_f(\x)$ can be found in \cite{ver2016}, and as shown in \cite{Ver2015},
 the Gaussian mean widths of subsets of the unit sphere of
the form $\T= \mathcal{C}_f(\x)  \cap \mathcal{S}^{n-1}$
plays a significant role in structured recovery, and
$\omega^2 ( \mathcal{C}_f(\x)  \cap \mathcal{S}^{n-1})$ is an approximation for the dimension of $\mathcal{K}$ \cite{Vershynin2020}.

\begin{definition}[Phase Transition Function]\label{m0-1}\cite[Definition 4]{oymak2018}
Let $\mathcal{C}_f(\x) $ be a cone of descent of $f$ at $\x$. Set $\omega=\omega ( \mathcal{C}_f(\x)  \cap \mathcal{B}^n)$, let 
 $\phi(t)=\sqrt{2}\frac{\Gamma(\frac{t+1}{2})}{\Gamma(\frac{t}{2})}\approx \sqrt{t}$, then the phase transition function is defined as 
\begin{equation*}
 \M(f,\x,u)=\phi^{-1}(\omega+u)\approx (\omega+u)^2  
\end{equation*}
where \textcolor{blue}{ $u>0$ is a parameter controlling
	the probability of success.}
\end{definition}
We use the short hand $m_0=\M(f,\x,u)$ with the dependence on $f$, $\x$, $u$ implied. 
It was shown in \cite{oymak2018} that for convex $f$, $m_0$ is exactly the minimum number
of measurements required for the program \eqref{LS} to succeed
in recovering the $\x^*$ with high probability. 
Besides, we define $\M(f,\x)$,
which approximately characterizes the
minimum number of samples required.
\begin{definition}\label{m0-2}
Let $\mathcal{C}_f(\x) $ be a cone of descent of $f$ at $\x$. Set $\omega=\omega ( \mathcal{C}_f(\x)  \cap \mathcal{S}^
{n-1})$, then the approximate minimum sample function is defined as 
$
\M(f,\x)=\omega^2.
$
\end{definition}
We denote $m_1=\M(f,\x)$ for convenience.
Other similar settings of the minimal number of samples. 
can be found in \cite[Definition 3]{Mahdi}.
The estimations of  \eqref{LS}, \eqref{LAD} and \eqref{non-LS} rely on the size of the feasible set $\mathcal{K}$, therefore we provide upper
bounds based on the Gaussian width to quantify the complexity of the general subset $\K\subset \mathbb{R}^n$.
The smaller this cone is the more suited the
function $f$ is at capturing the properties of $\x^*$. 
The Gaussian width serves as a valuable tool to assess the size of the feasible set $\mathcal{K}$,
offering the advantage of being robust to perturbations: if $\mathcal{K}$ experiences a slight increase, then 
the Gaussian width will only undergo a marginal change \cite{Plan2017}.
\textcolor{blue}{\begin{definition}
		A function $f:\mathbb{V}\rightarrow \mathbb{R}$ (where $\mathbb{V}$ is a vector space) is called absolutely homogeneous if,  it satisfies
		\[
		f(\alpha\x)=|\alpha| f(\x).
		\]
		for all scalars $\alpha\in \mathbb{R}$ and  vectors $\x\in \mathbb{V}$.
	\end{definition}
		For example, a norm is absolutely homogeneous.  The absolute homogeneity assumption of $f$ in the main theoretical results presented in Section \ref{sec:main}
		is necessary  in the case  where $0<f(\x^*)<\eta$ due to its role in enabling the use of the intermediate variable $\x^{\eta}$.
		Future work could investigate whether an analogous version of \eqref{re2} (in Theorem \ref{th1} ) can be obtained without imposing further assumptions.
}

\textcolor{blue}{
In rest of this section, we first present several lemmas relevant to the random process, followed by a characterization of the approximate isotropy of Gaussian random matrices with Gaussian width. Then we provide lemmas regarding the deviation of random matrices on sets, which establish two-sided bounds on Gaussian random matrices.}

\begin{lemma}\label{lem1}\cite[Theorem 27]{Mahdi}
    Suppose that $A\in \mathbb{R}^{m\times n}$ is a Gaussian random matrix with independent $\bma_i\sim \mathcal{N}(0,I_n)$ rows, then for any subset $\T\subset \mathbb{R}^n$ and any $u\geq 0$, $\delta\in (0,1)$, the event
    \[
    \left |\frac{\|A\x\|_2}{\phi(m)}-\|\x\|_2
    \right |\leq \delta \|\x\|_2
    \]
    holds for all $\x\in \T$ with probability at least $1-2\exp(-\frac{u^2}{2})$ as long as
    $
    m\geq \frac{(\omega(\T)+u)^2}{\delta^2}.
    $
\end{lemma}
\begin{lemma}\label{lem1-0}\cite[Corollary 1.5]{liaw2017}
Under the assumptions of Lemma \ref{lem1}, for any $u\geq 0$ the event
\[
\sup_{\x\in \mathcal{T}\cap\mathcal{S}^{n-1}}
\left|\|A\x\|_2-\sqrt{m}\right|
\leq C[\omega(\mathcal{T}\cap\mathcal{S}^{n-1})+u],
\]
holds with probability at least \textcolor{blue}{$1-2\exp(-u^2)$}.
\end{lemma}
The next lemma provides another version of Lemma \ref{lem1} and Lemma \ref{lem1-0} described by $m_0$ defined in Definition \ref{m0-1}.
\begin{lemma}\label{lem1-1}\cite[Lemma 7.5]{Mahdi2017}
    Let $\C=\C_f(\x)$ and $m_0=\M(f,\x,u)$, suppose that $A\in \mathbb{R}^{m\times n}$ is a Gaussian random matrix with independent $\bma_i\sim \mathcal{N}(0,I_n)$ rows, then the following holds with probability at least $1-2\exp(-u^2)$:
    \[
    \inf \limits_{\z\in\C\cap\mathcal{S}^{n-1}}\|A\z\|_2^2\geq (\phi(m)-\phi(m_0))^2.
    \]
\end{lemma}

\begin{lemma}\label{m0-3}\cite[Lemma 23]{oymak2018}
    Define $\phi(t)=\sqrt{2}\frac{\Gamma(\frac{t+1}{2})}{\Gamma(\frac{t}{2})}\approx \sqrt{t}$ as in Definition \ref{m0-1}, then for $0\leq m_0\leq m$ we have
    \[
    \frac{\phi(m_0)}{\sqrt{m_0}}\leq  \frac{\phi(m)}{\sqrt{m}}.
    \]
\end{lemma}

\begin{lemma}\label{innerproduct}
     Suppose that $A\in \mathbb{R}^{m\times n}$ is a Gaussian random matrix with independent $\bma_i\sim \mathcal{N}(0,I_n)$ rows, then for any subset $\T\subset \mathcal{B}_2^n$ and any $u\geq 0$, the event 
   \[   \sup\limits_{\x\in\T,\e\in\mathbb{R}^m}|\langle \x,A^T\e\rangle|\leq \|\e\|_2[\omega(\T)+u]
   \]
   holds with probability at least $1-2\exp(-Cu^2)$. 
\end{lemma}
\textcolor{blue}{The proof of Lemma \ref{innerproduct} can be found in the Appendix.}

\begin{lemma}\label{lem4}
Under the conditions of Lemma \ref{lem1},
suppose that $\bar{\x}\in\mathbb{R}^n$ is a fixed vector, then
  \[
  \min\limits_{\Omega\subset [m],\ |\Omega|\geq m/2}\|A_{\Omega}\bar{\x}\|_2\geq \frac{v_0}{2}\sqrt{m}\|\bar{\x}\|_2
  \]
  holds with probability 
  $1-2\exp\left(-\frac{v_0^2}{8}m\right)$ and $v_0=\frac{1}{32e}\sqrt{\frac{\pi}{2}}
  \left(1-\frac{1}{4\sqrt{\pi}}\right)\approx 0.0124$.
\end{lemma}
 \begin{proof}
 The proof of Lemma \ref{lem4} is similar to the proof of \cite[Lemma 4.4]{xu2016}.
 Without loss of generality, we assume that $\|\bar{\x}\|_2=1$. And set $\z=A\bar{\x}$,~$\z_{\Omega}=A_{\Omega}\bar{\x}$, thus the entries of $\z$ are independent realizations of Gaussian random variables $\z_i\sim \mathcal{N}(0,1)$. Apply \cite[Lemma 4.2]{xu2016} with $m\geq 1$ and \cite[Lemma 4.3]{xu2016} by setting $t=\frac{v_0}{2}$, we have
 \[
 \mathbb{P} \left[\min\limits_{\Omega\subset [m],\ |\Omega|\geq m/2}\frac{1}{m}\|\z_{\Omega}\|_2^2\geq \frac{v_0^2}{4}\right]\geq 1-2\exp\left(-\frac{v_0^2}{8}m\right).
 \]
 \end{proof}

\section{Main Results}\label{sec:main}
\textcolor{blue}
{This section presents our main results.
To aid our analysis, we categorize the range of possible values.
Each case displays unique performance characteristics, requiring varied estimation methods.}

\subsection{Linear Estimation}
\begin{theorem}\label{th1}
Let $\x^*$ be an arbitrary vector in $\mathbb{R}^n$. For any $u>0$,
let $m_0=\M(f,\x^*,u)$, which is defined in Definition \ref{m0-1}.
Let $A\in \mathbb{R}^{m\times n}$ be a Gaussian random matrix with independent $\bma_i\sim \mathcal{N}(0,I_n)$ rows.
Suppose that $\hat{\x}\in \mathbb{R}^n$ is a solution to
the \textcolor{blue}{constrained} least square \eqref{LS}
with $\y=A\x^*+\bm{e}$.
If 
\[
m\gtrsim m_0,
\]
then 
\begin{equation}\label{re1}
\|\hat{\x}-\x^*\|_2\leq \frac{3\sqrt{2}}{2(1-\rho)}\|\bm{P}_{\mathcal{K}}(\x^*)-\x^*\|_2+\left(\frac{4\sqrt{2}\rho}{3(1-\rho)}+\frac{4\rho}{(1-\rho)^2}\right)\frac{\|\e\|_2}{\sqrt{m}},
\quad 
\text{with}\quad  
\rho\gtrsim \sqrt{\frac{m_o}{m}},
\end{equation}
\textcolor{blue}{  
holds for all $f(\x^*)\geq\eta$  with probability at least $1-6\exp(-u^2)$.}
Suppose that $0<f(\x^*)<\eta$ and $f$ is absolutely \textcolor{blue}{homogeneous}. 
Denote 
\[
\x^{\eta}=\frac{\eta}{f(\x^*)}\x^*,
\quad 
m'_0=\M(f,\x^{\eta},u),
\]
If 
\[
m\gtrsim m_0',
\]
then for any $u>0$, the condition
\begin{equation}\label{re2}
  \|\hat{\x}-\x^*\|_2\leq \left(6\frac{\rho'}{(1-\rho')^2}+1\right)\left(\frac{\eta}{f(\x^*)}-1\right)\|\x^*\|_2+\frac{4\rho'}{(1-\rho')^2}\frac{\|\e\|_2}{\sqrt{m}},
  \quad 
  \text{with}\quad  
  \rho'\gtrsim \sqrt{\frac{m_o}{m}},
\end{equation}
holds with probability at least $1-6\exp(-u^2)$.
\end{theorem}
\begin{proof}
\textbf{Case 1: $f(\x^*)\geq\eta$}. 
Since $\hat{\x}$ is a  solution to \eqref{LS}, we have
\begin{equation}\label{eq:41:1}
\|A\hat{\x}-\y\|_2^2\leq \|A\bm{P}_{\mathcal{K}}(\x^*)-\y\|_2^2.
\end{equation}
We set 
\[
\h=\hat{\x}-\x^*=\h_1+\h_2
\]
where 
\[
\h_1=\hat{\x}-\bm{P}_{\mathcal{K}}(\x^*) 
\quad 
\text{and} 
\quad 
\h_2=\bm{P}_{\mathcal{K}}(\x^*)-\x^*.
\]
It follows from \eqref{eq:41:1}
\textcolor{blue}{
\begin{align*}
\|A\h\|_2^2-2\langle A^T\e,\h\rangle+\|\e\|_2^2
&=\|A\hat{\x}-A\x^*-\e\|_2^2\\
&=\|A\hat{\x}-\y\|_2^2\\
&\leq \|A\bm{P}_{\mathcal{K}}(\x^*)-\y\|_2^2\\
&=\|A\bm{P}_{\mathcal{K}}(\x^*)-A\x^*-\e\|_2^2\\
&=\|A\h_2-\e\|_2^2 \\
&=\|A\h_2\|_2^2-2\langle A^T\e,\h_2\rangle+\|\e\|_2^2.
\end{align*}}
Therefore, 
\begin{equation}\label{eq1}
  \|A\h\|_2^2\leq \|A\h_2\|_2^2+2\langle A^T\e,\h-\h_2\rangle=\|A\h_2\|_2^2+2\langle A^T\e,\h_1\rangle.
\end{equation}
Since $\h\in\mathcal{C}_f(\x^*)$, the normalized vector $\h/\|\h\|_2$ lies in the set
$
T=\mathcal{C}_f(\x^*)\cap \mathcal{S}^{n-1}
$. 
Therefore, applying Lemma \ref{lem1-1}, we get
\begin{equation}\label{eq2}
\|A\h\|_2^2\geq (\phi(m)-\phi(m_0))^2\|\h\|_2^2
\end{equation}
with probability at least $1-2\exp(-u^2)$.

Next, \textcolor{blue}{since $f(\h_2+\x^*)=f(\bm{P}_{\mathcal{K}}(\x^*))\leq \eta<f(\x^*)$, we have $\h_2\in \mathcal{C}_f(\x^*)$. }
Taking $\x=\h_2/\|\h_2\|_2$ in  Lemma \ref{lem1-0}, we get that
\begin{equation}\label{eq:41:2}
\left|\|A\h_2\|_2-\sqrt{m}\|\h_2\|_2\right|\leq C(\omega(\T)+u)\|\h_2\|_2
\leq C\sqrt{m_0}\|\h_2\|_2  
\leq \frac{\sqrt{m}}{2}\|\h_2\|_2. 
\end{equation}
with probability at least $1-2\exp(-u^2)$.
It follows from \eqref{eq:41:2} that 
\begin{equation}\label{eq3}
  \|A\h_2\|_2^2\leq \frac{9m}{4}\|\h_2\|_2^2.
\end{equation}

Since the random process $\langle A^T\e,\h_1\rangle=\langle A^T\e,\h-\h_2\rangle$, both $\h$ and $\h_2$ satisfy the condition in Lemma \ref{innerproduct}, with probability at least $1-2\exp(-Cu^2)$. 
It follows from  $\omega(\T)+u\leq \sqrt{m_0}$ that
\begin{align}\label{eq4}
  \langle A^T\e,\h_1\rangle 
  =\langle A^T\e,\h-\h_2\rangle 
  \leq \|\e\|_2(\omega(\T)+u)(\|\h\|_2+\|\h_2\|_2)\notag
  \leq \sqrt{m_0}\|\e\|_2(\|\h\|_2+\|\h_2\|_2).
\end{align}
Combining the above result with \eqref{eq1}, \eqref{eq2}, and \eqref{eq3}, we find
\begin{align*}
\left(\phi(m)-\phi(m_0)\right)^2\|\h\|_2^2&
\leq \|A\h\|_2^2\leq \|A\h_2\|_2^2+2\langle A^T\e, \h_1\rangle \\
&\leq 
\frac{9m}{4}\|\h_2\|_2^2+2\sqrt{m_0}\|\e\|_2(\|\h\|_2+\|\h_2\|_2).
\end{align*}
Since
 Lemma \ref{m0-3} yields
\[
1-\rho \leq 1-\sqrt{\frac{m_0}{m} }\leq 1-\frac{\phi(m_0)}{\phi(m)},
\]
we have
\begin{align*}
  (1-\rho)^2\|\h\|_2^2&\leq 
\frac{9m\|\h_2\|_2^2}{4\phi^2(m)}+
\color{blue}{\frac{2\sqrt{m_0}\|\e\|_2(\|\h\|_2+\|\h_2\|_2)}{\phi^2(m)}}\\
&\leq \frac{9}{2}\|\h_2\|_2^2+4\rho(\|\h\|_2+\|\h_2\|_2)\frac{\|\e\|_2}{\sqrt{m}}
\end{align*}
where the last inequality holds by using $\phi^2(m)\geq \frac{1.0049}{2}m$. Therefore,
\begin{align*}
&\left ((1-\rho)\|\h\|_2-\frac{2\rho}{1-\rho}\frac{\|\e\|_2}{\sqrt{m}}\right)^2 \\
\leq&
\left(\frac{3\sqrt{2}}{2}\|\h_2\|_2+\frac{2\sqrt{2}\rho}{3}\frac{\|\e\|_2}{\sqrt{m}}\right)^2
+
\left(\frac{2\rho}{1-\rho}\frac{\|\e\|_2}{\sqrt{m}}\right)^2-
\left(\frac{2\sqrt{2}\rho}{3}\frac{\|\e\|_2}{\sqrt{m}}\right)^2\\
\leq& \left(\frac{3\sqrt{2}}{2}\|\h_2\|_2+
\left(\frac{4\sqrt{2}\rho}{3}+\frac{2\rho}{1-\rho}\right)\frac{\|\e\|_2}{\sqrt{m}}
\right)^2.
\end{align*}
Taking square root both sides of the above inequality and rearranging
 the terms, we have
\begin{equation*}
  \|\h\|_2\leq \frac{3\sqrt{2}}{2(1-\rho)}\|\h_2\|_2+
\left(\frac{4\sqrt{2}\rho}{3(1-\rho)}+\frac{2\rho}{(1-\rho)^2}\right)\frac{\|\e\|_2}{\sqrt{m}}.
\end{equation*}

\textbf{Case 2: $0<f(\x^*)<\eta$}.
 Since $f$ is absolutely \textcolor{blue}{homogeneous},   we have
$
f(\x^{\eta})=  \eta \tfrac{f(\x^*)}{f(\x^*)}=\eta
$

\textcolor{blue}{where}
$\x^{\eta}=\frac{\eta} {f(\x^*)}\x^*$. \textcolor{blue}{Denoting}
$\tilde{\e}=A\x^*-A\x^{\eta}+\e$
, we have
\[
\y=A\x^*+\e=A\x^{\eta}+(A\x^*-A\x^{\eta}+\e)=A\x^{\eta}+\tilde{\e}.
\]
\textcolor{blue}{Similar to Case 1,} we set $\h=\hat{\x}-\x^*=\h_1-\h_2$, where $\h_1=\hat{\x}-\x^{\eta}$ and $\h_2=\x^*-\x^{\eta}$.

Since $\hat{\x}$ is a  solution to \eqref{LS}, we have
\begin{equation}\label{eqcase2-1}
\|A\hat{\x}-\y\|_2^2\leq \|A\x^{\eta}-\y\|_2^2.
\end{equation}
It follows from \eqref{eqcase2-1}
\textcolor{blue}{
\begin{align*}
\|A\h_1\|_2^2-2\langle A^T\tilde{\e},\h_1\rangle+\|\tilde{\e}\|_2^2&=\|A\h_1-\tilde{\e}\|_2^2\\
&=\|A\hat{\x}-A\x^{\eta}-\tilde{\e}\|_2^2\\
&=\|A\hat{\x}-\y\|_2^2\\
&
\leq \|A\x^{\eta}-\y\|_2^2=\|\tilde{\e}\|_2^2.
\end{align*}}
Therefore, 
\begin{equation}\label{eq2-1}
  \|A\h_1\|_2^2\leq 2\langle A^T\tilde{\e},\h_1\rangle.
\end{equation}
Since $\h_1\in\mathcal{C}_f(\x^{\eta})$, the normalized vector $\h_1/\|\h_1\|_2$ lies in the set
$
T'=\mathcal{C}_f(\x^{\eta})\cap \mathcal{S}^{n-1}
$. Therefore, applying Lemma \ref{lem1-1}, we get 
\begin{equation}\label{eq2-2}
  \|A\h_1\|_2^2\geq (\phi(m)-\phi(m'_0))^2\|\h_1\|_2^2.
\end{equation}

Next, by taking $\x=\h_1/\|\h_1\|_2$ in Lemma \ref{innerproduct}, we have
\begin{equation}\label{eq2-3}
  \langle A^T\tilde{\e},\h_1\rangle\leq \sqrt{m'_0}\|\tilde{\e}\|_2\|\h_1\|_2.
\end{equation}
Moreover, since \textcolor{blue}{$\h_2/\|\h_2\|_2$}
 also lies in the set $\T'$,
we get that with probability at least  $1-2\exp(-u^2)$,
\begin{align}\label{eq2-4}
  \|\tilde{\e}\|_2 & \leq \|A\h_2\|_2+\|\e\|_2 
  \leq \left(C(\omega(\T')+u)+\sqrt{m}\right)\|\h_2\|_2+\|\e\|_2 \notag \\
   & \leq \frac{3}{2}\sqrt{m}\|\h_2\|_2+\|\e\|_2 
   =\frac{3}{2}\sqrt{m}\|\x^{\eta}-\x^*\|_2+\|\e\|_2,
\end{align}
where, in the second inequality, we used Lemma \ref{lem1-0}, and in the third inequality, we used the assumption on $m$.
Combining \eqref{eq2-1}, \eqref{eq2-2}, \eqref{eq2-3} and \eqref{eq2-4}, we find 
\begin{align*}
  (\phi(m)-\phi(m'_0))^2\|\h_1\|^2_2&\leq
  \|A\h_1\|_2^2\leq 2\langle A^T\tilde{\e},\h_1\rangle
  \leq \color{blue}{2\sqrt{m'_0}\|\tilde{\e}\|_2\|\h_1\|_2}\\
  &\leq 2\sqrt{m'_0}\|\h_1\|_2\left(\frac{3}{2}\sqrt{m}\|\x^{\eta}-\x^*\|_2+\|\e\|_2\right).
\end{align*}
Finally, using Lemma \ref{m0-3} and
 $\phi^2(m)\geq \frac{1.0049}{2}m$, we have
\begin{equation}\label{eq2-5}
\|\h_1\|_2\leq 
\frac{6\rho'}{(1-\rho')^2}\|\x^{\eta}-\x^*\|_2+\frac{4\rho'}{(1-\rho')^2}\frac{\|\e\|_2}{\sqrt{m}}.
\end{equation}
Concluding, 
we obtain \eqref{re2} by \textcolor{blue}{\eqref{eq2-5}
\begin{align*}
\|\h\|_2
&= 
\| \h_1-\h_2\|_2 \\
&\leq \|\h_1\|_2+\|\x^{\eta}-\x^*\|_2 
 \leq
\left(\frac{6\rho'}{(1-\rho')^2}+1\right)\|\x^{\eta}-\x^*\|_2+\frac{4\rho'}{(1-\rho')^2}\frac{\|\e\|_2}{\sqrt{m}}.
\end{align*}}
\end{proof}

\begin{remark}
 \textcolor{blue}
 {Theorem \ref{th1} implies that  both \eqref{re1} and \eqref{re2} are bounded by 
$C\rho\frac{\|\e\|_2}{\sqrt{m}}$  when $f(\x^*)=\eta$. }
The constant $C>1$ in the definition 
$\rho\geq C\sqrt{\frac{m_0}{m}}~(\rho\gtrsim \sqrt{\frac{m_0}{m}}) $
is consistent with that in the condition 
$m\geq Cm_0~(m\gtrsim m_0)$, which ensures the denominator $1-\rho>0$.
\end{remark}
\textcolor{blue}{\begin{remark}
	Both the parameter $\rho$ and $\rho'$ represent the proportionality constant between number of measurements and signal structural complexity. If $f$ is a norm, then $\rho=\rho'$.
	It was showed in  \cite{oymak2018} the parameter $\rho$ and $\rho'$
	also determine the convergence rate of projection gradient algorithms.
	\end{remark}}

Theorem \ref{th1} suggests an interesting trade-off between sample complexity (number of measurements) and the \textcolor{blue}{distance between $f(\x^*)$ and the tuning parameter $\eta$ of the feasible set $\K$}. More precisely if one consider the exact recovery $\y=A\x^*$ and $m\geq Cm_0$, then the approximation error for $f(\x^*)\geq \eta$ is of the form
\[
\|\hat{\x}-\x^*\|_2\leq 
\frac{3\sqrt{2}}{2\left(1-\sqrt{\frac{1}{C}}\right)}
\|\bm{P}_{\mathcal{K}}(\x^*)-\x^*\|_2.
\]
This expression shows that for a given $\x^*$ and $\K$ (fixed structural complexity), the larger $m$ or data complexity is (compared with the minimal number of measurements $m_0$), the smaller the approximation error is.
\textcolor{blue}
{The term
\[
\frac{3\sqrt{2}}{2(1-\rho)}\|\bm{P}_{\mathcal{K}}(\x^*)-\x^*\|_2
\]
in \eqref{re1} and the term
\[
\left(6\frac{\rho'}{(1-\rho')^2}+1\right)\left(\frac{\eta}{f(\x^*)}-1\right)\|\x^*\|_2
\]
in \eqref{re2} both originate from the  the mismatch between the tuning parameter $\eta$ and $f(\x^*)$. }
This ``mismatch error'' goes to zero as $\eta\rightarrow f(\x^*)$ and a larger number of measurements lead to a smaller mismatch error \cite{oymak2018}.
Moreover, for fixed $m$, the closer $f(\x^*)$ to the feasible set $\K$, the smaller the approximation error is. \textcolor{blue}{Note that the ``mismatch'' discussed in this paper differs from the ``mismatch covariance'' in \cite{genzel2018mismatch}, which is defined as the covariance between $\x$ and the residual $\y-A\x$. Here we specifically quantify the discrepancy between the parameter $\eta$ and the true signal structure function $f$.}

\textcolor{blue}{In \cite[Theorem V.1 (Asymptotic Singularity)]{Berk2022}, the authors investigated the asymptotic singularity of the constrained LASSO problem
	\begin{equation}\label{eta-lasso}
		\min\limits_{\|\x\|_1 \leq \eta} \|\y - A\x\|_2^2.
	\end{equation}
	They assume that $\x^*$  is $s$-sparse, and that $\y=A\x^*+\sigma\e$,  where $\e\sim \mathcal{N}(0,I_m)$ is noise and $\sigma>0$. Given $\eta>0$, let $\hat{\x}$ denote the solution of \eqref{eta-lasso}, 
	If
	 $\|\x^*\|_1\neq\eta$ and $m\ge C s\log (n/s)$, then 
	\begin{equation}\label{eq:low}
		\frac{\| \hat{\x} - \x^* \|_2^2}{\sigma^{2}} = \infty, 
	\end{equation}
	holds almost surely.
	We note that  the   asymptotic  analysis \eqref{eq:low}  depends on the lower bound of $\| \hat{\x} - \x^* \|_2^2$, whereas we establish non-asymptotic upper error bounds of $\| \hat{\x} - \x^* \|_2^2$ in Theorem \ref{th1}.}

\textcolor{blue}
{We derive several corollaries from Theorem \ref{th1}, which characterize the results when $f$ is specified as a given function.
The  first result is on the  standard Lasso.}
\begin{corollary}\label{co1}
Let $\x^*$ be
an  $s$-sparse vector in $\mathbb{R}^n$.
Suppose that $A\in \mathbb{R}^{m\times n}$ is a Gaussian random matrix with independent $\bma_i\sim \mathcal{N}(0,I_n)$ rows.
Let $\hat{\x} \in\mathbb{R}^n$ be a solution to 
\begin{equation}\label{L1-LS}
    \min_{\|\x\|_1\leq \|\x^*\|_1} \frac{1}{2}\|\y-A\x\|_2^2
\end{equation}
where 
$\y=A\x^*+\bm{e}$.
If
$
 m\gtrsim s\log n,$
then 
\begin{equation}\label{eq:lasso}
\|\hat{\x}-\x^*\|_2\lesssim \sqrt{\frac{s\log n}{m}}\frac{\|\e\|_2}{\sqrt{m}}
\end{equation}
holds with probability at least $1-2\exp(-s\log n)$.
\end{corollary}
\begin{proof}
	\textcolor{blue}{
Since $\eta=\|\x^*\|_1$,
 the mismatch term in \eqref{re1} vanishes.
The effective dimension $m_0$ is controlled by $s\log n$
according to the fact in \cite[Lemma 10.5.3]{Vershynin2020} that any vector in $\mathcal{C}_f(\x^*)\cap \mathcal{S}^{n-1}$ where $\mathcal{C}_f(\x^*)$ is the closed cone of $\mathcal{D}_f(\x^*)$ with
\[
\mathcal{D}_f(\x^*)=\{\h:~\|\x^*+\h\|_1\leq \|\x^*\|_1\}.
\]
lies in the set
\begin{equation}\label{l1set}
	\T_1^s=\{\x\in\mathbb{R}^n:\|\x\|_1\leq \sqrt{s},\|\x\|_2\leq 1\}.
\end{equation}
  Then
the estimation \eqref{eq:lasso} can be
obtained 
from Theorem \ref{th1}
by setting $u=2\sqrt{s\log n}$ the assumption that $m\gtrsim s\log n$ with $\rho$ is proportional to $\sqrt{\frac{m_0}{m}}=\sqrt{\frac{s\log n}{m}}$, and the denominator $(1-\rho)^2$ can be seen as some constant in $(0,1)$, thus we finish our proof. }
\end{proof}

 The second result is on the  least square with $l_2$ norm constraint.
\begin{corollary}\label{co3}
Let $\x^*$ be a vector in $\mathbb{R}^n$.
 Suppose that $A\in \mathbb{R}^{m\times n}$ is a Gaussian random matrix with independent $\bma_i\sim \mathcal{N}(0,I_n)$ rows.
Let $\hat{\x}$ be a solution to 
\begin{equation}\label{L2-LS}
    \min_{\|\x\|_2\leq \|\x^*\|_2} \frac{1}{2}\|\y-A\x\|_2^2,
\end{equation}
where $\y=A\x^*+\bm{e}$.
 If
 $ 
 m\gtrsim n,
$
then  
\begin{equation}
\|\hat{\x}-\x^*\|_2\lesssim \sqrt{\frac{n}{m}}\frac{\|\e\|_2}{\sqrt{m}}
\end{equation}
holds with probability at least $1-2\exp(-Cn)$.
\end{corollary} 
\begin{proof}
Similar to the proof of Corollary \ref{co1}, 
\textcolor{blue}{
	Since $\eta=\|\x^*\|_2$,
	the mismatch term in \eqref{re1} vanishes.
	In this case, $m_0\approx (\omega(\T)+u)^2\approx n$ with  $u=\sqrt{n}$.
Noticed that when $m\gtrsim n$, we have $\rho=C\sqrt{\frac{ n}{m}}$, and}
 the denominator $(1-\rho)^2$ can be seen as some constant in $(0,1)$. \textcolor{blue}{This completes the proof.}
\end{proof}

We 
\textcolor{blue}{observe
that 
Corollary \ref{co1} has been established in \cite{Oymak2013} and 
\cite[Theorem 10.6.1]{Vershynin2020}, while
the result in Corollary \ref{co3} 
presents
the same upper bound for linear estimation as  established in \cite{Plan2017}.}

It can be inferred from the Corollary \ref{co1} that the analog of the classical compressed sensing result is included in our result as the special case $f(\x^*)=\|\x^*\|_1=\eta$.
A consistent conclusion with Corollary \ref{co1} can be achieved
when it comes to non-convex constraints 
\[
\K=\{\x\in\mathbb{R}^n:\ \|\x\|_p\leq \|\x^*\|_p= \eta\},\quad 0\leq p<1
\]
based on the following property of the Gaussian width: $\omega(\T)=\omega(\conv(\T))$.
\textcolor{blue}{We assert that Theorem \ref{th1} admits multiple proof approaches. In particular, an alternative proof can be derived by leveraging Lemma \ref{lem1}, yielding a result characterized by a multiplicative factor expressed via $\delta$ as follows.
\begin{theorem}\label{th1-c}
	Let $\x^*$ be an arbitrary vector in $\mathbb{R}^n$. For any $u>0$,
	let $m_0=\M(f,\x^*,u)$, which is defined in Definition \ref{m0-1}.
	Let $A\in \mathbb{R}^{m\times n}$ be a Gaussian random matrix with independent $\bma_i\sim \mathcal{N}(0,I_n)$ rows.
	Suppose that $\hat{\x}\in \mathbb{R}^n$ is a solution to
	the constrained least square \eqref{LS}
	with $\y=A\x^*+\bm{e}$.
	If 
	\[
	m\geq \frac{1}{\delta^2} m_0,\quad \text{with} \quad \delta\in (0,1)
	\]
	then 
	\begin{equation*}
		\|\hat{\x}-\x^*\|_2\leq \frac{\sqrt{2}(1+\delta)}{(1-\rho)}\|\bm{P}_{\mathcal{K}}(\x^*)-\x^*\|_2+\left(\frac{4\sqrt{2}\rho}{3(1-\rho)}+\frac{4\rho}{(1-\rho)^2}\right)\frac{\|\e\|_2}{\sqrt{m}},
		\quad 
		\text{with}\quad  
		\rho=\frac{1}{\delta} \sqrt{\frac{m_o}{m}},
	\end{equation*}
		holds for all $f(\x^*)\geq\eta$  with probability at least $1-6\exp(-u^2)$.
	Suppose that $0<f(\x^*)<\eta$ and $f$ is absolutely homogeneous. 
	Denote 
	\[
	\x^{\eta}=\frac{\eta}{f(\x^*)}\x^*,
	\quad 
	m'_0=\M(f,\x^{\eta},u),
	\]
	If 
	\[
	m\geq \frac{1}{\delta^2} m_0',\quad \text{with}\quad \delta\in(0,1), 
	\]
	then for any $u>0$, the condition
	\begin{equation*}
		\|\hat{\x}-\x^*\|_2\leq \left(4\frac{\rho'(1+\delta)}{(1-\rho')^2}+1\right)\left(\frac{\eta}{f(\x^*)}-1\right)\|\x^*\|_2+\frac{4\rho'}{(1-\rho')^2}\frac{\|\e\|_2}{\sqrt{m}},
		\quad 
		\text{with}\quad  
		\rho=\frac{1}{\delta} \sqrt{\frac{m_o}{m}},
	\end{equation*}
	holds with probability at least $1-6\exp(-u^2)$.
\end{theorem}
}
The proof of Theorem \ref{th1-c} can be found in the Appendix.
In the following we provide the estimation of the solution \textcolor{blue}{to} the \textcolor{blue}{constrained} LAD  \eqref{LAD}.
\begin{theorem}\label{th1-2}
Let $\x^*$ be an arbitrary vector in $\mathbb{R}^n$.
\textcolor{blue}{
Let $m_1=\M(f,\x^*)$, which is defined in Definition \ref{m0-2}}.
Let $\gamma\geq4$, $\beta>(\frac{5}{4}\gamma)^2$. Suppose that $A\in \mathbb{R}^{m\times n}$ is a Gaussian random matrix with independent $\bma_i\sim \mathcal{N}(0,I_n)$ rows.
  Let $\hat{\x}$ be a solution to the 
  \textcolor{blue}{constrained} LAD
  \eqref{LAD} with $\y=A\x^*+\bm{e}$. 
\textcolor{blue}{If
\begin{equation}\label{lad-con1}
	m>\beta m_1,
\end{equation} 
then for $0<u<\sqrt{\frac{2}{\pi}}-\rho$,
\begin{equation}\label{lad-re1}
	\|\hat{\x}-\x^*\|_2\leq 
	\frac{
		\sqrt{\frac{2}{\pi}}+\rho+u}{\sqrt{\frac{2}{\pi}}-\rho-u}
	\|
	\bm{P}_{\mathcal{K}}(\x^*)
	-
	\x^*
	\|_2
	+
	\frac{2}{\sqrt{\frac{2}{\pi}}-\rho-u}\frac{\|\e\|_1}{m},
    \quad \text{with}\quad \rho=\gamma\sqrt{\frac{m_1}{m}},
\end{equation}
holds for all $f(\x^*)\geq\eta$
with probability at least $1-6\exp(-\frac{mu^2}{2})$.
}  
Suppose that $0<f(\x^*)<\eta$ and $f$ is absolutely \textcolor{blue}{homogeneous}. 
Denote 
\begin{equation*}
\x^{\eta}=\frac{\eta}{f(\x^*)}\x^*,
\quad 
m'_1=\M(f,\x^{\eta}).
\end{equation*}
If 
\begin{equation}\label{lad-con2}
 m>\beta m'_1,
\end{equation}
then for any $0<u<\sqrt{\frac{2}{\pi}}-\rho'$,
\begin{equation}\label{lad-re2}
\|\hat{\x}-\x^*\|_2\leq 
\frac{3\sqrt{\frac{2}{\pi}}+\rho'+u}{\sqrt{\frac{2}{\pi}}-\rho'-u}
\left(\frac{\eta}{f(\x^*)}-1\right)
\|\x^*\|_2+\frac{2}{\sqrt{\frac{2}{\pi}}-\rho'-u}\frac{\|\e\|_1}{m}, \quad \text{with}\quad 
\rho'=\gamma\sqrt{\frac{m'_1}{m}},
\end{equation}
holds with probability at least $1-6\exp(-\frac{mu^2}{2})$.
\end{theorem}
The proof of Theorem \ref{th1-2} can be found in Appendix.
Theorem \ref{th1-2} also suggests 
tradeoffs
between the ``mismatch error'' and the data complexity. 
We focus on the case \textcolor{blue}{when} $f(\x^*)<\eta$. Let $\epsilon$ be the desired relative accuracy of the optimal solution, then we have
\[
\frac{\|\hat{\x}-\x^*\|_2}{\|\x^*\|_2}\leq \epsilon=
\frac{3\sqrt{\frac{2}{\pi}}+\rho'+u}{\sqrt{\frac{2}{\pi}}-\rho'-u}
\left(\frac{\eta}{f(\x^*)}-1\right)
+\frac{2}{\sqrt{\frac{2}{\pi}}-\rho'-u}\frac{\|\e\|_1}{m\|\x^*\|_2}.
\]
One can observe that a larger number of measurements leads to smaller values of the rate $\rho'$, 
which in turn leads to a smaller mismatch error 
\[
\frac{3\sqrt{\frac{2}{\pi}}+\rho'+u}{\sqrt{\frac{2}{\pi}}-\rho'-u}
\left(\frac{\eta}{f(\x^*)}-1\right).
\]
On the other hand, 
when the noise level $\|\e\|_1$ is smaller enough compared with $\|\x^*\|_2$, the  inequality 
\[
\frac{3\sqrt{\frac{2}{\pi}}+\rho'+u}{\sqrt{\frac{2}{\pi}}-\rho'-u}
\left(\frac{\eta}{f(\x^*)}-1\right)\leq \epsilon
\]
holds. Simple manipulations imply that
\begin{equation}\label{mis}
m\geq m_0\frac{(\mismatch+\epsilon)^2}{(C_2\epsilon-C_1\cdot \mismatch)^2},
\end{equation}
where $C_1=3\sqrt{\frac{2}{\pi}}+u$, $C_2=\sqrt{\frac{2}{\pi}}-u$ and we use ``mismatch'' to represent the mismatch term $\frac{\eta}{f(\x^*)}-1$. The expression \eqref{mis} shows that less data is needed (smaller $m$ is required) if the term $\frac{\eta}{f(\x^*)}-1$ goes to zero.
Next, we consider the sparse case that $f(\x)=\|\x\|_1$.
\begin{corollary}\label{l1-lad}
Let  $\x^*$ be 
an arbitrary 
$s$-sparse vector in $\mathbb{R}^n$.
Let $\gamma\geq4$, $\beta>(\frac{5}{4}\gamma)^2$.
Suppose that $A\in \mathbb{R}^{m\times n}$ is a Gaussian random matrix with independent $\bma_i\sim \mathcal{N}(0,I_n)$ rows.
 Let $\hat{\x}$ be a solution to the constrained LAD:
 \begin{equation}\label{eq-l1-lad}
     \min\limits_{\|\x\|_1\leq \|\x^*\|_1}\|\y-A\x\|_1
 \end{equation}
 where $\y=A\x^*+\bm{e}$. 
If 
$$
m>\beta s\log n,
$$
then 
for $0<u<\sqrt{\frac{2}{\pi}}-\gamma\sqrt{\frac{s\log n}{m}}$,
\begin{equation}\label{eq:lad1}
\|\hat{\x}-\x^*\|_2\leq
\frac{2}{\sqrt{\frac{2}{\pi}}-u-\gamma\sqrt{\frac{s\log n}{m}}}\frac{\|\e\|_1}{m}
\end{equation}
holds
with probability at least $1-2\exp(-\frac{mu^2}{2})$.    
\end{corollary}
\begin{proof}
\textcolor{blue}
{Since $\eta=\|\x^*\|_1$, the ``mismatch error'' term vanishes. Since the constraint set is the same as in Corollary \ref{co1}, the effective dimension $m_1$ is also controlled by $s\log n$ and $\rho=\gamma \sqrt{\frac{m_1}{m}}\approx \sqrt{\frac{s\log n}{m}}$.}
Then, we obtain the conclusion \eqref{eq:lad1}.
\end{proof}

If the observations \textcolor{blue}{are}  corrupted  with adversarial corruption, i.e.,
\[
\y=A\x^*+\e_1+\e_2
\]
where $\|\x^*\|_0\leq s$, and $\|\e_1\|_0\leq \beta m$, \textcolor{blue}{then}  
it was shown in \cite{karma2019} that the constrained 
LAD \eqref{eq-l1-lad} is robust to any fraction of corruptions $\beta$ less than $\beta_0\approx 0.239$.
\textcolor{blue}{We note that only the case $\|\x^*\|_1 \le  \eta$ is studied in \cite{karma2019} and
the case $\|\x^*\|_1> \eta$ is not addressed.}
\textcolor{blue}{
Let $\hat{\x}$ be a solution to the following constrained LAD 
 \begin{equation*}
	\min\limits_{\|\x\|_1\leq \eta}\|\y-A\x\|_1.
\end{equation*}
Combing the methodology in the proof of \cite[Theorem 1.1]{karma2019}  with the technique in the proof of Theorem \ref{th1-2},
, we 
obtain that for $\|\x^*\|_1> \eta$, the following
\begin{equation}\label{eq:239lambda}
\|\hat{\x}-\x^*\|_2\lesssim \frac{1}{\epsilon-\frac{1}{\alpha}}\left(\frac{1}{m}\|\e_2\|_1+
\left(\sqrt{\frac{1}{2\pi}}+\frac{\epsilon}{2}\right)\|\mathcal{P}_{\K}(\x^*)-\x^*\|_2\right)+\frac{\|\x^*\|_1-\eta}{\alpha\sqrt{s}},
\end{equation}
holds  with a high probability as long as
$$
m
\gtrsim 
\frac{\alpha^2}{\epsilon^2}s\log
\left(\frac{en}{\alpha^2\epsilon s}\right),
$$
where 
$\epsilon>0$, $\beta<\beta_0-\epsilon$ with $\beta_0\approx 0.239$ and $\alpha\geq \frac{2}{\epsilon}$.
This estimation \eqref{eq:239lambda} can serve as a supplementary addition to  \cite[Theorem 1.1]{karma2019}.
}

\subsection{Phase Retrieval}
\begin{theorem}\label{th2}
Let $\x^*$ be an arbitrary vector in $\mathbb{R}^n$.
\textcolor{blue}{For any $u>0$,
let $m_0=\M(f,\x^*,u)$, which is defined in Definition \ref{m0-1}.
Suppose that $A\in \mathbb{R}^{m\times n}$ is a Gaussian random matrix with independent $\bma_i\sim \mathcal{N}(0,I_n)$ rows. 
Let $\hat{\x}$ be a solution to \eqref{non-LS} with $\y=|A\x^*|+\bm{e}$. 
If 
 \[
 m\gtrsim m_0,
 \]
then 
  \begin{equation}\label{th2-eq1}
  \min\{\|\hat{\x}-\x^*\|_2,\|\hat{\x}+\x^*\|_2\}\leq \left(\frac{4}{v_0}\rho+\frac{4}{v_0}+1\right)\|\bm{P}_{\mathcal{K}}(\x^*)
  - \x^*
  \|_2+\frac{4}{v_0}\frac{\|\e\|_2}{\sqrt{m}}, 
  \quad 
  \text{with}
  \quad 
  \rho\gtrsim\sqrt{\frac{m_0}{m}},
\end{equation}
holds   with probability at least 
$1-2\exp(-u^2)-2\exp
\left(\frac{v_0^2}{8}m\right)$ for all $f(\x^*)\geq\eta$.}
Suppose that
$f$ is absolutely \textcolor{blue}{homogeneous} and satisfies
$0<f(\x^*)<\eta$. 
\textcolor{blue}{Denote}
\[
m_0'=\M(f,\x^{\eta},u),
\quad 
v_0=\frac{1}{32e}\sqrt{\frac{\pi}{2}}
\left(1-\frac{1}{4\sqrt{\pi}}\right)\approx 0.0124.
\]
If
\[
m\gtrsim m'_0, 
\]
then for any $u>0$, 
\begin{align}\label{th2-eq2}
	\min\{\|\hat{\x}-\x^*\|_2,\|\hat{\x}+\x^*\|_2\} 
	\leq
	\left(\frac{24}{v_0^2}\rho'+\frac{3}{v_0}+1\right)\left(\frac{\eta}{f(\x^*)}-1\right)\|\x^*\|_2
	+
	\left(\frac{16}{v_0^2}\rho'+\frac{2}{v_0}\right)\frac{\|\e\|_2}{\sqrt{m}},   \quad \text{with}\quad \rho'\gtrsim\sqrt{\frac{m'_0}{m}},
\end{align}
 holds with  probability  at least 
$
1-2\exp(-u^2)-2\exp
\left(\frac{v_0^2}{8}m\right).
$
\end{theorem}
\begin{proof}
    \textbf{Case 1: $f(\x^*)\geq\eta$.}
Denote the index set
\[
  S_1  =\{j:\langle \bma_j,\hat{\x}\rangle\langle \bma_j,\bm{P}_{\mathcal{K}}(\x^*)\rangle >0\},\quad 
  S_2  =\{j:\langle \bma_j,\hat{\x}\rangle\langle \bma_j,\bm{P}_{\mathcal{K}}(\x^*)\rangle <0\}.
\]
Without loss of generality, we assume that $\sharp(S_1)=\beta m\geq m/2$~(otherwise, we can assume that
$\sharp(S_2)\geq m/2$).
\textcolor{blue}{
	We set $\h^-=\hat{\x}-\x^*$ and $\h^+=\hat{\x}+\x^*$.
Similar to the proof of Theorem \ref{th1}, denote }
\[
\h^-=\h_1+\h_2,
\]
where \[
\h_1=\hat{\x}-\bm{P}_{\mathcal{K}}(\x^*)
\quad
\textcolor{blue}
{\text{and}}
\quad
\h_2=\bm{P}_{\mathcal{K}}(\x^*)-\x^*.
\]
For $S_1$, we have
\begin{align}\label{non-eq1}
  \|A_{S_1}\h_1\|_2 & =\|A_{S_1}\hat{\x}- A_{S_1}\bm{P}_{\mathcal{K}}(\x^*)\|_2\notag\\
   &=\||A_{S_1}\hat{\x}|- |A_{S_1}\bm{P}_{\mathcal{K}}(\x^*)|\|_2 \notag \\
   &\leq \||A\hat{\x}|- |A\bm{P}_{\mathcal{K}}(\x^*)|\|_2\notag\\
   &\leq \||A\hat{\x}|-\y\|_2+\| |A\bm{P}_{\mathcal{K}}(\x^*)|-\y\|_2\notag\\
   &\leq 2\| |A\bm{P}_{\mathcal{K}}(\x^*)|-\y\|_2\notag\\
   &=2\| |A\bm{P}_{\mathcal{K}}(\x^*)|-|A\x^*|-\e\|_2\notag\\
   &\leq 2\| |A\bm{P}_{\mathcal{K}}(\x^*)|-|A\x^*|\|_2+2\|\e\|_2\notag\\
   &\leq 2\| A\h_2\|_2+2\|\e\|_2,
\end{align}
where the fifth inequality holds based on the fact that $\hat{\x}$ is the solution of \eqref{non-LS}.

Since \textcolor{blue}{$f(\h_2+\x^*)=f(\bm{P}_{\mathcal{K}}(\x^*))\leq \eta\leq f(\x^*)$,}
the normalized vector $\h_2/\|\h_2\|_2$ lies in the set
$T=\mathcal{C}_f(\x^*)\cap \mathcal{S}^{n-1}$, \textcolor{blue}{Then  Lemma \ref{lem1-0} implies that}
\begin{equation}\label{non-eq2}
  \|A\h_2\|_2\leq (C(\omega(\T)+u)+\sqrt{m})\|\h_2\|_2\leq C(\sqrt{m_0}+\sqrt{m})\|\h_2\|_2
\end{equation}
holds with probability at least $1-2\exp(-u^2)$.

In addition, Lemma \ref{lem4} implies that
\begin{equation}\label{non-eq3}
  \|A_{S_1}\h_1\|_2\geq \frac{v_0}{2}\sqrt{m}\|\h_1\|_2
\end{equation}
holds with probability at least $1-2\exp\left(-\frac{v_0^2}{8}m\right)$. 

Therefore, by \eqref{non-eq1}, \eqref{non-eq2} and \eqref{non-eq3}, we have
\begin{align*}
  \|\h^-\|_2 &\color{blue}{=\|\h_1+\h_2\|_2} \leq \|\h_1\|_2+\|\h_2\|_2 \\
  &\leq \frac{2}{v_0\sqrt{m}}\|A_{S_1}\h_1\|_2+\|\h_2\|_2 \\
   & \leq \frac{4}{v_0\sqrt{m}}(\|A\h_2\|_2+\|\e\|_2)+ \|\h_2\|_2\\
   & \leq \frac{4}{v_0\sqrt{m}}\left(C\sqrt{m_0}+\sqrt{m}\right) \|\h_2\|_2+\frac{4}{v_0\sqrt{m}}\|\e\|_2+\|\h_2\|_2 \\
   &\leq \left(\frac{4}{v_0}\rho+\frac{4}{v_0}+1\right)\|\h_2\|_2+\frac{4}{v_0}\frac{\|\e\|_2}{\sqrt{m}}
\end{align*}
holds with probability at least $1-2\exp(-u^2)-
2\exp
\left(\frac{v_0^2}{8}m\right)$.
For the case where $\sharp(S_2)\geq m/2$, we can conclude that
$\h^+$ shares
the same upper bound using a similar method.

\textbf{Case 2: $0<f(\x^*)<\eta$.}
We denote
$\x^{\eta}=\frac{\eta}{f(\x^*)}\x^*$ \textcolor{blue}{ and $
	\tilde{\e}=|A\x^*|-|A\x^{\eta}|+\e  
	$}. 
Since $f$ is absolutely \textcolor{blue}{homogeneous}, 
\textcolor{blue}{ we have}
\[
f(\x^{\eta})=\eta \frac{f(\x^*)}{f(\x^*)}=\eta.
\]
Furthermore,
\begin{equation*}
\y
=|A\x^*|+\e
=|A\x^{\eta}|+(|A\x^*|-|A\x^{\eta}|+\e)
=|A\x^{\eta}|+\tilde{\e},
\end{equation*}
where 
$\tilde{\e}$
can be viewed as the additive noise on the measurements $|A\x^{\eta}|$.

Since $\hat{\x}$ is the solution of \eqref{non-LS}, we have
\begin{equation}\label{phcase2-1}
\||A\hat{\x}|-\y\|_2^2\leq \||A\x^{\eta}|-\y\|_2^2.
\end{equation}
\textcolor{blue}{Different from Case 1, we }set
\begin{align*}
  S_1 & =\{j:{\rm sign}(\langle \bma_j,\hat{\x}\rangle)=1,{\rm sign}(\langle \bma_j,\x^{\eta}\rangle)=1\},\\
  S_2 & =\{j:{\rm sign}(\langle \bma_j,\hat{\x}\rangle)=-1,{\rm sign}(\langle \bma_j,\x^{\eta}\rangle)=-1\},\\
  S_3&=\{j:{\rm sign}(\langle \bma_j,\hat{\x}\rangle)=1,{\rm sign}(\langle \bma_j,\x^{\eta}\rangle)=-1\},\\
  S_4&=\{j:{\rm sign}(\langle \bma_j,\hat{\x}\rangle)=-1,{\rm sign}(\langle \bma_j,\x^{\eta}\rangle)=1\}.
\end{align*}
Without loss of generality, we assume that 
$|(S_1\cup S_2)|\geq m/2$. 
\textcolor{blue}{We also denote }
$$
\h^-=\hat{\x}-\x^*\quad \text{and} 
\quad  
\h^+=\hat{\x}+\x^*.
$$ 
\textcolor{blue}{In addition, we set
	\[
	\h^-=\h_1+\h_2,
	\]
	where
\[
\h_1=\hat{\x}-\x^{\eta}
\quad
\text{and}
\quad
\h_2=\x^{\eta}-\x^*.
\]
}
\textcolor{blue}{Then}
\[
\||A\hat{\x}|-\y\|_2^2
=\||A\hat{\x}|-|A\x^{\eta}|-\tilde{\e}\|_2^2
\geq \|A_{S_1}\h_1-\tilde{\e}_{S_1}\|_2^2+\|A_{S_2}\h_1+\tilde{\e}_{S_2}\|_2^2.
\]
This, together with \eqref{phcase2-1},  yields
\begin{equation}\label{non-eq41}
\|A_{S_1}\h_1-\tilde{\e}_{S_1}\|_2^2+\|A_{S_2}\h_1+\tilde{\e}_{S_2}\|_2^2\leq \|\tilde{\e}\|_2^2.
\end{equation}
\textcolor{blue}
{By expanding the squared expression in \eqref{non-eq41} and rearranging the terms, we derive}
\begin{align}\label{non-eq4}
\|A_{S_{12}}\h_1\|_2^2
\leq 2\langle \h_1,A_{S_1}^T\tilde{\e}_{S_1}-A_{S_2}^T\tilde{\e}_{S_2}\rangle+\|\tilde{\e}_{S_{12}^c}\|_2^2.
\end{align}
In addition, 
 Lemma \ref{lem4} implies that
 \begin{equation}\label{non-eq5}
   \|A_{S_{12}}\h_1\|_2^2\geq \frac{v_0^2}{4}m\|\h_1\|_2^2
 \end{equation}
holds  with probability at least 
$1-2\exp\left(-\frac{v_0^2}{8}m\right)$.

Next, we estimate \textcolor{blue}{the inner product term in \eqref{non-eq4}}
\[
\left\langle\h_1,A_{S_1}^T\tilde{\e}_{S_1}-A_{S_2}^T\tilde{\e}_{S_2}
\right\rangle.
\]
Since \textcolor{blue}{$f(\h_1+\x^{\eta})=f(\hat{\x})\leq \eta=f(\x^{\eta})$, then}
the normalized vector of $\h_1/\|\h_1\|_2$  lies in the set
 $\T'=\mathcal{C}_f(\x^{\eta})\cap \mathcal{S}^{n-1}$,
 by using Lemma \ref{innerproduct}, we have
\begin{equation}\label{non-eq6}
  \langle\h_1,A_{S_1}^T\tilde{\e}_{S_1}-A_{S_2}^T\tilde{\e}_{S_2}\rangle
  \leq 2(\omega(\T')+u)\|\tilde{\e}\|_2\|\h_1\|_2
  \leq 2\sqrt{m_0'}\|\tilde{\e}\|_2\|\h_1\|_2,
\end{equation}
 holds with probability at least $1-2\exp(-Cu^2)$. 
 
Substituting \eqref{non-eq5},
\eqref{non-eq6} into equation \eqref{non-eq4}, we have
\begin{align*}
\frac{v_0^2m}{4}\|\h_1\|_2^2
&\leq \|A_{S_{12}}\h_1\|_2^2\\
&\leq
2\langle \h_1,A_{S_1}^T\tilde{\e}_{S_1}-A_{S_2}^T\tilde{\e}_{S_2}\rangle+\|\tilde{\e}_{S_{12}^c}\|_2^2\\
&\leq 4\sqrt{m_0'}\|\tilde{\e}\|_2\|\h_1\|_2+\|\tilde{\e}_{S_{12}^c}\|_2^2,
\end{align*} 
\textcolor{blue}{By completing the square on both sides of the inequality, we obtain}
\begin{equation}\label{non-eq71}
\left(\frac{v_0}{2}\sqrt{m}\|\h_1\|_2-4\sqrt{\frac{m'_0}{m}}\frac{\|\tilde{\e}\|_2}{v_0}\right)^2\leq \|\tilde{\e}\|_2^2+\left(4\sqrt{\frac{m'_0}{m}}\frac{\|\tilde{\e}\|_2}{v_0}\right)^2
\leq \left(\left(v_0+4\sqrt{\frac{m'_0}{m}}\right)\frac{\|\tilde{\e}\|_2}{v_0}\right)^2,
\end{equation}
where the last inequality is implied by $a^2+b^2\leq (a+b)^2$.
 \textcolor{blue}{Taking square roots of both sides of \eqref{non-eq71} yields}
 \begin{equation}\label{non-eq7}
   \|\h_1\|_2\leq \frac{2}{v_0}\left(\frac{8}{v_0}\sqrt{\frac{m_0'}{m}}+1\right)\frac{\|\tilde{\e}\|_2}{\sqrt{m}}.
 \end{equation}
\textcolor{blue}{In the following, }we  estimate $\|\tilde{\e}\|_2$  by
\begin{align}\label{non-eq8-0}
   \|\tilde{\e}\|_2& =\||A\x^*|-|A\x^{\eta}|+\e\|_2 \notag\\
  & \leq \||A\x^*|-|A\x^{\eta}|\|_2+\|\e\|_2 \notag\\
   & \leq \|A(\x^*-\x^{\eta})\|_2+\|\e\|_2 
\end{align}
By the assumption on $m$ and \textcolor{blue}{taking 
$\x=\frac{\x^*-\x^{\eta}}{\|\x^*-\x^{\eta}\|_2}$ in
Lemma \ref{lem1-0}, we have
\begin{align}\label{non-eq8-1}
  \|A(\x^*-\x^{\eta})\|_2&\leq   \left(C(\omega(\T)+u)+\sqrt{m}\right)\|(\x^*-\x^{\eta})\|_2\notag\\
  &\leq \left(C\sqrt{m'_0}+\sqrt{m}\right)\|(\x^*-\x^{\eta})\|_2  
\leq \frac{3\sqrt{m}}{2}\|(\x^*-\x^{\eta})\|_2. 
\end{align}
Therefore, by \eqref{non-eq8-0} and \eqref{non-eq8-1}, we have
\begin{equation}\label{non-eq8}
    \|\tilde{\e}\|_2\leq \frac{3}{2}\sqrt{m}\|\x^*-\x^{\eta}\|_2+\|\e\|_2.
\end{equation}
}
Substituting \eqref{non-eq8} into \eqref{non-eq7}, we obtain
\begin{align}\label{non-eq8final}
\|\h_1\|_2&\leq
\left(\frac{24}{v_0^2}\sqrt{\frac{m_0'}{m}}+\frac{3}{v_0}\right)\|\x^*-\x^{\eta}\|_2+
\left(\frac{16}{v_0^2}\sqrt{\frac{m_0'}{m}}+\frac{2}{v_0}\right)\frac{\|\e\|_2}{\sqrt{m}}\notag\\
&=\left(\frac{24}{v_0^2}\sqrt{\frac{m_0'}{m}}+\frac{3}{v_0}\right)\left(\frac{\eta}{f(\x^*)}-1\right)\|\x^*\|_2+
\left(\frac{16}{v_0^2}\sqrt{\frac{m_0'}{m}}+\frac{2}{v_0}\right)\frac{\|\e\|_2}{\sqrt{m}}.
\end{align}
\textcolor{blue}{
The proof is completed by \eqref{non-eq8final} and the triangle inequality
\begin{align*}
\|\hat{\x}-\x^*\|_2&=\|\h^-\|_2\leq \|\h_1\|_2+\|\h_2\|_2\\
&=\|\hat{\x}-\x^{\eta}\|_2+\|\x^{\eta}-\x^*\|_2\\
&\leq \left(\frac{24}{v_0^2}\sqrt{\frac{m_0'}{m}}+\frac{3}{v_0}+1\right)\left(\frac{\eta}{f(\x^*)}-1\right)\|\x^*\|_2+
\left(\frac{16}{v_0^2}\sqrt{\frac{m_0'}{m}}+\frac{2}{v_0}\right)\frac{\|\e\|_2}{\sqrt{m}}.
\end{align*}}
\textcolor{blue}{For the case where $|(S_3\cup S_4)|\geq m/2$, we can obtain the same bound for $\|\h^+\|_2$ by a similar method to above.}
\end{proof}

\textcolor{blue}{
When the parameters are chosen optimally, we obtain the following two corollaries regarding the sparse phase retrieval problem.}

\begin{corollary}\label{co-pr}
Let $\x^*$ be an arbitrary $s$-sparse vector in  $\mathbb{R}^n$.
Suppose that $A\in \mathbb{R}^{m\times n}$ is a Gaussian random matrix with independent $\bma_i\sim \mathcal{N}(0,I_n)$ rows.
Let $\hat{\x}\in \mathbb{R}^n$ be a solution to 
 \begin{equation}\label{l1-phase}
  \min_{\|\x\|_1\leq \|\x^*\|_1} \frac{1}{2}
  \left\|\y-|A\x|\right\|_2^2
\end{equation}
where $\y=|A\x^*|+\bm{e}$. If
$$
m\gtrsim s\log n,
$$
then 
\[
\min\{\|\hat{\x}-\x^*\|_2,\|\hat{\x}+\x^*\|_2\}\lesssim
\frac{\|\e\|_2}{\sqrt{m}}
\]
holds with
probability at least $1-3\exp(-Cm)$.
\end{corollary}

\begin{proof}
Since $\eta=\|\x^*\|_1$, \textcolor{blue}{the ``mismatch error'' terms in \eqref{th2-eq1} and \eqref{th2-eq2} vanish. Similar to the proof in Corollary \ref{co1},
the effective dimension $m_0$} is controlled by $s\log n$. 
Notice that when $m\gtrsim s\log n$, the value of $\rho'$ lies in the interval $(0,1)$, thus both the terms
$\frac{4}{v_0}$ in \eqref{th2-eq1} and $\frac{16}{v^2_0}\rho'+\frac{2}{v_0}$ in \eqref{th2-eq2} can be regraded as some constant.
Then we obtain from \textcolor{blue}{Theorem \ref{th2} that
\[\min\{\|\hat{\x}-\x^*\|_2,\|\hat{\x}+\x^*\|_2\}
\leq \frac{C}{v_0}\frac{\|\e\|_2}{\sqrt{m}}.
\]}
 \end{proof}
\noindent Corollary \ref{co-pr} is consistent with
\cite[Theorem I.5]{xu2020}, and the bound $\frac{\|\e\|_2}{\sqrt{m}}$ is proved to be \textcolor{blue}{sharp in \cite[Remark I.4]{xu2020}}.
\textcolor{blue}{\begin{corollary}\label{coL0-pr}
		Let $\x^*$ be an arbitrary $s$-sparse vector in  $\mathbb{R}^n$.
		Suppose that $A\in \mathbb{R}^{m\times n}$ is a Gaussian random matrix with independent $\bma_i\sim \mathcal{N}(0,I_n)$ rows.
		Let $\hat{\x}\in \mathbb{R}^n$ be a solution to 
		\begin{equation}\label{l0-phase}
			\min_{\|\x\|_0\leq  s} \frac{1}{2}
			\left\|\y-|A\x|\right\|_2^2
		\end{equation}
		where $\y=|A\x^*|+\bm{e}$. If
		$$
		m\gtrsim s\log n,
		$$
		then 
		\[
		\min\{\|\hat{\x}-\x^*\|_2,\|\hat{\x}+\x^*\|_2\}\lesssim
		\frac{\|\e\|_2}{\sqrt{m}}
		\]
		holds with
		probability at least $1-3\exp(-Cm)$.
\end{corollary}
\begin{proof}
	This proof can be directly established by combining the following fact
	that 
	$
	\omega(\T_1^s)\leq 2\omega(\T_0^s)\leq C\sqrt{s \log n}
	$
	where 
	$
	\T_0^s=\{\x:~\|\x\|_0\leq s,\|\x\|_2\leq 1\}$, and the proof of Corollary \ref{co-pr}.
	\end{proof}}
\textcolor{blue}{
\begin{remark}
	For sparse phase retrieval problem \eqref{l0-phase}, iterative algorithms such as Iterative Hard Thresholding \cite{cai2022}, Truncated Amplitude Flow \cite{wang2017}, CoPRAM \cite{Jagatap2019} typically consist of two stages: initialization and refinement. The initialization stage employs spectral initialization, requiring $\mathcal{O}(s^2\log n)$
	Gaussian samples to obtain a sufficiently accurate estimate. In the refinement stage, various algorithms further optimize the estimate, most achieving linear convergence with only $\mathcal{O}(s\log n)$
	samples. Thus, the total sample complexity is dominated by the initialization stage. 
	The gap between the model's  recovery guarantee with $\mathcal{O}(s\log n)$
measurements and the algorithm's actual sample complexity (inflated by initialization) constitutes a key challenge to be addressed in subsequent work.
\end{remark}
}

\textcolor{blue}
{Comparing Theorem \ref{th2} with  Corollary \ref{co-pr} we see that there are extra terms 
$$
\left(\frac{4}{v_0}\rho+\frac{4}{v_0}+1\right)\|\bm{P}_{\mathcal{K}}(\x^*)
- \x^*\|_2,\quad \text{for} \quad  f(\x^*)\geq \eta,
$$ 
and
$$
\left(\frac{24}{v_0^2}\rho'+\frac{3}{v_0}+1\right)\left(\frac{\eta}{f(\x^*)}-1\right)\|\x^*\|_2, 
\quad \text{for} \quad
0<f(\x^*)<\eta.
$$}
The extra terms are aligned with those in the linear case, resulting from the mismatch between the tuning parameter $\eta$ and $f(\x^*)$. 
Theorem \ref{th2} demonstrates that as we increase the data complexity $m$, the rate $\rho$ or $\rho'$
decreases, leading to a reduction in mismatch error.
Theorem \ref{th2} also illustrates the trade-off between the ``mismatch error'' and the complexity of the data.
To demonstrate this, let us consider the scenario where $f(\x^*)\geq\eta$, we use 
$\epsilon$ to denote the upper bound of the distance between $\hat{\x}$ and $\x^*$:
\[
\min\{\|\hat{\x}-\x^*\|_2,\|\hat{\x}+\x^*\|_2\}\leq \epsilon=\left(\frac{4}{v_0}\rho+\frac{4}{v_0}+1\right)
\|
\bm{P}_{\mathcal{K}}(\x^*)
- \x^*
\|_2+\frac{4}{v_0}\frac{\|\e\|_2}{\sqrt{m}}.
\]
One can observe that a larger number of measurements leads to smaller values of the rate $\rho$, 
which in turn leads to a smaller mismatch error 
\[
\left(\frac{4}{v_0}\rho+\frac{4}{v_0}+1\right)
\|
\bm{P}_{\mathcal{K}}(\x^*)-\x^*
\|_2.
\]
Meanwhile, since for given $\epsilon$, the following inequality holds:
\[
\left(\frac{4}{v_0}\rho+\frac{4}{v_0}+1\right)\|\bm{P}_{\mathcal{K}}(\x^*)
- \x^*\|_2\leq \epsilon
\]
by some simple computation, we have
\begin{equation}\label{dis1}
    m\geq m_0\frac{16(\mismatch)^2}{(\epsilon v_0-(4+v_0)\mismatch)^2}
\end{equation}
where we use ``mismatch'' to denote $\|\bm{P}_{\mathcal{K}}(\x^*)
- \x^*\|_2$. Therefore, we can deduce that
less $m$ is required if the 
mismatch term $\|\bm{P}_{\mathcal{K}}(\x^*)
- \x^*\|_2$ goes to zero by formula \eqref{dis1}, which is similar to the linear estimation case.
\begin{remark}
    Analogous to the context of linear estimation, by specifying a specific value for $C$ in the condition $m\geq Cm_0$, we can derive theorems that are comparable to those presented in Theorem \ref{th1-c}, pertaining specifically to  Theorem \ref{th2}. \textcolor{blue}{Parallel to the case in linear estimation, analogous parameter sensitivity analyses exist for adversarial phase retrieval via nonlinear least absolute deviation under heavy-tailed noise \cite{huang2023}, which is a natural generalization we omit here.}
\end{remark}

\section*{Acknowledgments}
The authors wish to express their thanks to Professor 
Zhiqiang Xu for helpful discussions regarding related theories. The authors also would like to thank the referees for valuable comments.

\bibliographystyle{IEEEtran}

\bibliography{references}

\section*{APPENDIX}
\textcolor{blue}{
\begin{lemma}\label{concentration}
	\cite[Gaussian Concentration Inequality]{Vershynin2020}
	Consider a random vector $\x\sim \mathcal{N}(0,I_n)$ and a Lipschitz function $g:\mathbb{R}^n\rightarrow \mathbb{R}$ with constant $\|g\|_{L}$:
	\[
	|g(\x)-g(\y)|\leq \|g\|_{L}\|\x-\y\|_2.
	\]
	Then for every $t>0$, the event
	$
	|g(\x)-\mathbb{E}g(\x)|\leq t
	$
	holds with probability at least
	$1-2\exp\left(-\tfrac{Ct^2}{\|g\|^2_{L}}\right)$.
	\end{lemma}
Similar to Lemma \ref{lem1}, we provide the following concentration inequalities for $\|A\x\|_1$, which play a significant role in the estimation of Theorem \ref{th1-2}.
\begin{lemma}\label{L1-lemma1}\cite[Lemma 2.1]{plan2014}
	Consider a bounded subset $\T\subset \mathbb{R}^n$ and independent random vectors $\bma_i\sim \mathcal{N}(0,I_n)$, $i=1,2,\ldots,m$. Let
	\[
	Z=\sup\limits_{\x\in\T}
	\left|
	\frac{1}{m}\sum\limits_{i=1}^m|\langle \bma_i,\x\rangle|-\sqrt{\frac{2}{\pi}}\|\x\|_2
	\right|.
	\]
\begin{enumerate}
		\item One has $\mathbb{E}(Z)\leq \frac{4\omega(\T)}{\sqrt{m}}$.
		\item The following deviation inequality holds for all $u>0$
		\[
		\mathbb{P}
		\left[
		Z>\frac{4\omega(\T)}{\sqrt{m}}+u
		\right]
		\leq 2\exp\left(-\frac{mu^2}{2\rad^2(\T)}\right).
		\]
\end{enumerate}
\end{lemma}
In the proof of Theorem \ref{th1-2}, we state the result of Lemma \ref{L1-lemma1} in terms of Gaussian random matrix (under the assumption of Lemma \ref{lem1}) \cite[Remark 2.2]{plan2014}, where the event $Z$ is expressed as
\[
\sup\limits_{\x\in\T}
\left|
\frac{1}{m}\|A\x\|_1-\sqrt{\frac{2}{\pi}}\|\x\|_2
\right|.
\]
}
\paragraph*{Proof of Lemma \ref{innerproduct}}
\begin{proof}
    For any fixed $\e\in\mathbb{R}^m$, we have
    \[
\mathbb{E}\sup\limits_{\x\in\T}\langle\x,A^T\e\rangle=\|\e\|_2\cdot\omega(\T),
    \]
    which follows from the definition of the Gaussian width. Next, we set
    \[   g(A):=\sup\limits_{\x\in\T}\langle\x,A^T\e\rangle.
    \]
    For any matrix $A,~B\in \mathbb{R}^{m\times n}$, we have
    \begin{align*}
        \left |\sup\limits_{\x\in\T}\langle\x,A^T\e\rangle-\sup\limits_{\x\in\T}\langle\x,B^T\e\rangle
        \right |&\leq
        \left|
\sup\limits_{\x\in\T}\langle(A-B)\x,\e\rangle \right|\\
&\leq \|\e\|_2\|A-B\|_F ,
    \end{align*}
   we use the fact that $\x\in\T\subset \mathcal{B}_2^n$ in the last inequality. Then by applying Lemma \ref{concentration}, we have
   \[
\sup\limits_{\x\in\T}\langle\x,A^T\e\rangle\leq    \mathbb{E}\sup\limits_{\x\in\T}\langle\x,A^T\e\rangle+t
   \]
   with \textcolor{blue}{ probability at least}
 $1-2\exp
 \left(-\frac{Ct^2}{\|g\|^2_{L}}\right)$. Choosing $t=\|\e\|_2u$ where $u>0$ is arbitrary, we obtain 
 \[ \sup\limits_{\x\in\T}\langle\x,A^T\e\rangle\leq \|\e\|_2(\omega(\T)+
u) \]
 with \textcolor{blue}{ probability at least}
 $1-2\exp(-Cu^2)$.
\end{proof}
\paragraph*{Proof of Theorem \ref{th1-c}}
\begin{proof}
	While the proof of this theorem exhibits strong similarities with Theorem \ref{th1}, it primarily differs in certain aspects of bounding $\|A\h\|_2$.
	
\textbf{Case 1: $f(\x^*)\geq \eta$}. 
Since $\hat{\x}$ is a solution to \eqref{LS}, we have the same formulas in \eqref{eq1} and \eqref{eq2}:
\begin{equation*}
  \|A\h\|_2^2\leq \|A\h_2\|_2^2+2\langle A^T\e,\h_1\rangle,
\end{equation*}
and
\[
\|A\h\|_2^2\geq (\phi(m)-\phi(m_0))^2\|\h\|_2^2.
\]
We set $\h=\hat{\x}-\x^*=\h_1+\h_2$, where $\h_1=\hat{\x}-\bm{P}_{\mathcal{K}}(\x^*)$ and $\h_2=\bm{P}_{\mathcal{K}}(\x^*)-\x^*$.
Different from the proof in Theorem \ref{th1}, we next use Lemma \ref{lem1} to bound $\|A\h_2\|_2$.
By  Lemma \ref{lem1} we get that with probability at least $1-2\exp(-u^2)$:
\textcolor{blue}{\[
\left|\frac{\|A\h_2\|_2}{\phi(m)}-\|\h_2\|_2\right|\leq \delta\|\h_2\|_2
\]
holds, 
it yields
\begin{equation}\label{ceq3}
  \|A\h_2\|_2^2\leq (1+\delta)^2m\|\h_2\|_2^2,
\end{equation}
by using the} inequality  $\phi(m)\leq \sqrt{m}$.
\textcolor{blue}{Since $f(\h+\x^*)=f(\hat{\x})\leq \eta\leq f(\x^*)$ and 
$f(\h_2+\x^*)=f(\bm{P}_{\mathcal{K}}(\x^*))\leq \eta\leq f(\x^*)$, both $\frac{\h}{\|\h\|_2}$ and $\frac{\h_2}{\|\h_2\|_2}$ lie in the set
$\mathcal{C}_f(\x^*)\cap \mathcal{S}^{n-1}$.}
\textcolor{blue}{Note that the random process $\langle A^T\e,\h_1\rangle =\langle A^T\e,\h-\h_2\rangle$,  and both $\h$ and $\h_2$ satisfy the condition in Lemma \ref{innerproduct}, then for any $u>0$,
\begin{align}\label{ceq1}
\langle A^T\e,\h_1\rangle &=\langle A^T\e,\h-\h_2\rangle\notag\\
&\leq [\omega(\mathcal{C}_f(\x^*)\cap \mathcal{S}^{n-1})+u]\|\e\|_2(\|\h\|_2+\|\h_2\|_2)\notag\\
&\leq \sqrt{m_0}\|\e\|_2(\|\h\|_2+\|\h_2\|_2)
\end{align}
holds with probability at least $1-2\exp(-Cu^2)$.} 
\textcolor{blue}{Combining with \eqref{eq1}, \eqref{eq2}, \eqref{ceq1} and \eqref{ceq3},} we find
\begin{align*}
	\left(\phi(m)-\phi(m_0)\right)^2\|\h\|_2^2&
	\leq \|A\h\|_2^2\leq \|A\h_2\|_2^2+2\langle A^T\e, \h_1\rangle \\
	&\leq 
	(1+\delta)^2m\|\h_2\|_2^2+2\sqrt{m_0}\|\e\|_2(\|\h\|_2+\|\h_2\|_2).
\end{align*}

Then by computation, we obtain
\begin{align*}
	(1-\rho)^2\|\h\|_2^2&\leq 
	\frac{(1+\delta)^2m\|\h_2\|_2^2}{\phi^2(m)}+
	\frac{2\sqrt{m_0}\|\e\|_2(\|\h\|_2+\|\h_2\|_2)}{\phi^2(m)}\\
	&\leq 2(1+\delta)^2\|\h_2\|_2^2+4\rho(\|\h\|_2+\|\h_2\|_2)\frac{\|\e\|_2}{\sqrt{m}}
\end{align*}
where the last inequality holds by using $\phi^2(m)\geq \frac{1.0049}{2}m$. Therefore,
\begin{align*}
	&\left ((1-\rho)\|\h\|_2-\frac{2\rho}{1-\rho}\frac{\|\e\|_2}{\sqrt{m}}\right)^2 \\
	\leq&
	\left(\sqrt{2}(1+\delta)\|\h_2\|_2+\frac{2\sqrt{2}\rho}{3}\frac{\|\e\|_2}{\sqrt{m}}\right)^2
	+
	\left(\frac{2\rho}{1-\rho}\frac{\|\e\|_2}{\sqrt{m}}\right)^2-
	\left(\frac{2\sqrt{2}\rho}{3}\frac{\|\e\|_2}{\sqrt{m}}\right)^2\\
	\leq& \left(\sqrt{2}(1+\delta)\|\h_2\|_2+
	\left(\frac{4\sqrt{2}\rho}{3}+\frac{2\rho}{1-\rho}\right)\frac{\|\e\|_2}{\sqrt{m}}
	\right)^2.
\end{align*}
Taking square root both sides of the above inequality and rearranging
the terms, we have
\begin{equation*}
	\|\h\|_2\leq \frac{\sqrt{2}(1+\delta)}{(1-\rho)}\|\h_2\|_2+
	\left(\frac{4\sqrt{2}\rho}{3(1-\rho)}+\frac{2\rho}{(1-\rho)^2}\right)\frac{\|\e\|_2}{\sqrt{m}}.
\end{equation*}
\textbf{Case 2: $0<f(\x^*)<\eta$}.
Since $f$ is absolutely \textcolor{blue}{homogeneous},
we have
\[
f(\x^{\eta})=\eta\frac{f(\x^*)}{f(\x^*)}
=\eta.
\]
\textcolor{blue}{where} $\x^{\eta}=\frac{\eta}{f(\x^*)}\x^*$ 
\textcolor{blue}{Denoting
$\tilde{\e}=A\x^*-A\x^{\eta}+\e$ and 
\[
\h=\hat{x}-\x^*=\h_1-\h_2
\]
where $\h_1=\hat{\x}-\x^{\eta}$ and $\h_2=\x^*-\x^{\eta}$}
, we have the following relation as in \eqref{eq2-1} in a similar way:
\begin{equation}\label{ceq2-1}
  \|A\h_1\|_2^2\leq 2\langle A^T\tilde{\e},\h_1\rangle.
\end{equation}

Similar to the proof of Theorem \ref{th1}(case 2),
the inner product $\langle A^T\tilde{\e},\h_1\rangle$ and $\|\tilde{\e}\|_2$ are also bounded as in \eqref{eq2-3}:
\begin{equation*}
  \langle A^T\tilde{\e},\h_1\rangle\leq \sqrt{m'_0}\|\tilde{\e}\|_2\|\h_1\|_2.
\end{equation*}
\textcolor{blue}{In addition, note that $\frac{\h_2}{\|\h_2\|_2}$ lies in the set
$\mathcal{C}_{f}(\x^{\eta})\cap \mathcal{S}^{n-1}$, then taking $\x=\frac{\h_2}{\|\h_2\|_2}$ in Lemma \ref{lem1}, we obtain
 \begin{align}\label{ceq2-4}
 	\|\tilde{\e}\|_2 & \leq \|A\h_2\|_2+\|\e\|_2 
 	\leq \left(\phi(m)(1+\delta)\right)\|\h_2\|_2+\|\e\|_2 \notag \\
 	& \leq (1+\delta)\sqrt{m}\|\h_2\|_2+\|\e\|_2 
 	= (1+\delta)\sqrt{m}\|\x^{\eta}-\x^*\|_2+\|\e\|_2,
 \end{align}
}
Combining \eqref{ceq2-1}, \eqref{eq2-2}, \eqref{eq2-3} and \eqref{ceq2-4}, we get 
\textcolor{blue}{
\begin{align*}
(\phi(m)-\phi(m'_0))^2\|\h_1\|^2_2&\leq
\|A\h_1\|_2^2\leq 2\langle A^T\tilde{\e},\h_1\rangle
\leq 2\sqrt{m'_0}\|\tilde{\e}\|_2\|\h_1\|_2\\
&\leq 2\sqrt{m'_0}\|\h_1\|_2\left((1+\delta)\sqrt{m}\|\x^{\eta}-\x^*\|_2+\|\e\|_2\right),
\end{align*}}
This, together with Lemma \ref{m0-3} and
$\phi^2(m)\geq \frac{1.0049}{2}m$ leads to
\begin{equation}\label{ceq22}
	\|\h_1\|_2\leq 
	\frac{4\rho'(1+\delta)}{(1-\rho)^2}\|\x^{\eta}-\x^*\|_2+\frac{4\rho'}{(1-\rho)^2}\frac{\|\e\|_2}{\sqrt{m}}.
\end{equation}
Concluding, \textcolor{blue}{by \eqref{ceq22} and the triangle inequality, we have}
\begin{align*}
	\|\h\|_2
	&= 
	\| \h_1-\h_2\|_2 \\
	&\leq \|\h_1\|_2+\|\x^{\eta}-\x^*\|_2 
	\leq
	\left(\frac{4\rho'(1+\delta)}{(1-\rho')^2}+1\right)\|\x^{\eta}-\x^*\|_2+\frac{4\rho'}{(1-\rho')^2}\frac{\|\e\|_2}{\sqrt{m}}.
\end{align*}
    
\end{proof}

\paragraph*{Proof of Theorem \ref{th1-2}}
\textbf{Case 1: \textcolor{blue}{$f(\x^*)\geq\eta$}}. 
Since $\hat{\x}$ is a solution to \eqref{LAD}, we have
\begin{equation}\label{th46-eq1}
\|A\hat{\x}-\y\|_1\leq \|A\bm{P}_{\mathcal{K}}(\x^*)-\y\|_1.
\end{equation}
We set 
\[
\h=\hat{\x}-\x^*=\h_1+\h_2
\]
where 
\begin{equation}\label{eq:h1}
\h_1=\hat{\x}-\bm{P}_{\mathcal{K}}(\x^*)
\quad 
\text{and}
\quad \h_2=\bm{P}_{\mathcal{K}}(\x^*)-\x^*.
\end{equation}
It follows from 
\eqref{th46-eq1} and \eqref{eq:h1} that
\textcolor{blue}{
\begin{align*}
\|A\h-\e\|_1&= \|A\hat{\x}-A\x^*-\e\|_1\notag\\ &=\|A\hat{\x}-\y\|_1\notag\\
 &\leq \|A\bm{P}_{\mathcal{K}}(\x^*)-\y\|_1\notag\\
 &=\|A\bm{P}_{\mathcal{K}}(\x^*)-A\x^*-\e\|_1\notag\\
 &\leq \|A\h_2\|_1+\|\e\|_1.
\end{align*}}
Therefore,
\begin{equation}\label{lad-eq1}
  \|A\h\|_1\leq  \|A\h-\e\|_1+\|\e\|_1\leq  \|A\h_2\|_1+2\|\e\|_1.
\end{equation}
Since the normalized error $\h/\|\h\|_2$ lies in the set
$\T=\mathcal{C}_f(\x^*)\cap \mathcal{S}^{n-1}$,
by taking $\x=\h/\|\h\|_2$ in 
Lemma \ref{L1-lemma1}, we have \textcolor{blue}{for any $u>0$,}
\begin{equation}\label{th46-eq2}
\sup\limits_{\x=\h/\|\h\|_2}\left|\frac{1}{m}\|A\x\|_1-\sqrt{\frac{2}{\pi}}\right|  \leq \frac{4\omega(\T)}{\sqrt{m}}+u
\leq \gamma \sqrt{\frac{m_1}{m}}+u
= \rho+u
\end{equation}
holds with probability at least $1-2\exp(-mu^2/2)$.
It follows from \eqref{th46-eq2} \textcolor{blue}{that}
\begin{equation}\label{lad-eq2}
  \left(\sqrt{\frac{2}{\pi}}-\rho-u\right)\|\h\|_2\leq \frac{1}{m}\|A\h\|_1.
\end{equation}

Next, \textcolor{blue}{we provide a bound of $\|A\h_2\|_1$ in a similar way}. Since the normalized vector $\h_2/\|\h_2\|_2$ also lies in the set $\T$, then
 by taking $\x=\h_2/\|\h_2\|_2$ in Lemma \ref{L1-lemma1}, we have the following inequality holds with probability at least $1-2\exp(-mu^2/2)$:
\begin{equation*}
\sup\limits_{\x=\h_2/\|\h_2\|_2}
\left|\frac{1}{m}\|A\x\|_1-\sqrt{\frac{2}{\pi}}\right|  \leq \frac{4\omega(\T)}{\sqrt{m}}+u\leq 
\rho+u,
\end{equation*}
which yields
\begin{equation}\label{lad-eq3}
  \frac{1}{m}\|A\h_2\|_1\leq \left(
  \sqrt{\frac{2}{\pi}}+\rho+u
  \right)
  \|\h_2\|_2
\end{equation}
Combining \eqref{lad-eq1}, \eqref{lad-eq2} and \eqref{lad-eq3}, we have
\begin{align*}
\left(\sqrt{\frac{2}{\pi}}-\rho-u\right)\|\h\|_2
&\leq \frac{1}{m}\|A\h\|_1 \\
&\leq \frac{1}{m}\|A\h_2\|_1+\frac{2}{m}\|\e\|_1 \\
&\leq
\left(\sqrt{\frac{2}{\pi}}+\rho+u\right)\|\h_2\|_2+\frac{2}{m}\|\e\|_1.
\end{align*}
Therefore, by computation, we have
\begin{equation*}
\|\hat{\x}-\x^*\|_2\leq 
\frac{
\sqrt{\frac{2}{\pi}}+\rho+u}{\sqrt{\frac{2}{\pi}}-\rho-u}
\|
\bm{P}_{\mathcal{K}}(\x^*)
-
\x^*
\|_2
+
\frac{2}{\sqrt{\frac{2}{\pi}}-\rho-u}\frac{\|\e\|_1}{m}.
\end{equation*}

\textbf{Case 2: $0<f(\x^*)<\eta$}.
Since $f$ is absolutely \textcolor{blue}{homogeneous}, we have
\begin{equation*}
f(\x^{\eta})=  \eta \frac{f(\x^*)}{f(\x^*)}=\eta,
\end{equation*}
\textcolor{blue}{where } $\x^{\eta}=\frac{\eta}{f(\x^*)}\x^*$. \textcolor{blue}{Denoting}
$\tilde{\e}=A\x^*-A\x^{\eta}+\e$, we have
\begin{equation*}
\y=A\x^*+\e=A\x^{\eta}+(A\x^*-A\x^{\eta}+\e)=A\x^{\eta}+\tilde{\e}.
\end{equation*}
We set $\h=\hat{\x}-\x^*=\h_1-\h_2$, 
where $\h_1=\hat{\x}-\x^{\eta}$ and $\h_2=\x^*-\x^{\eta}$.
 Since $\hat{\x}$ is the solution of \eqref{LAD}, the following inequality holds
\[
\|A\hat{\x}-\y\|_1\leq \|A\x^{\eta}-\y\|_1,
\]
which yields to
\textcolor{blue}{
\begin{align*}
\|A\h_1-\tilde{\e}\|_1&=\|A\hat{\x}-A\x^{\eta}-\tilde{\e}\|_1\notag\\
&=\|A\hat{\x}-\y\|_1\\
&\leq \|A\x^{\eta}-\y\|_1\notag\\
&=\|\tilde{\e}\|_1.
\end{align*}}
Therefore,
\begin{equation}\label{lad-eq2-1}
    \|A\h_1\|_1\leq \|A\h_1-\tilde{\e}\|_1+\|\tilde{\e}\|_1\leq 
    2\|\tilde{\e}\|_1.
\end{equation}
Since the normalized error $\h_1/\|\h_1\|_2$ lies in the set
$
\T'=\mathcal{C}_f(\x^{\eta})\cap \mathcal{S}^{n-1}
$,
\textcolor{blue}{by taking }
$\x=\h_1/\|\h_1\|_2$ in 
Lemma \ref{L1-lemma1}, 
\textcolor{blue}{we conclude that for any $u>0$}
\begin{equation}\label{th46-eq3}
\sup\limits_{\x=\h_1/\|\h_1\|_2}\left|\frac{1}{m}\|A\x\|_1-\sqrt{\frac{2}{\pi}}\right|  \leq \frac{4\omega(\T)}{\sqrt{m}}+u
\leq \gamma \sqrt{\frac{m'_1}{m}}+u
= \rho'+u
\end{equation}
holds with probability at least $1-2\exp(-mu^2/2)$.
It follows from \eqref{th46-eq3} \textcolor{red}{that}
\begin{equation}\label{lad-eq2-2}
 \left(
 \sqrt{\frac{2}{\pi}}-\rho'-u
 \right)
 \|\h_1\|_2\leq \frac{1}{m}\|A\h_1\|_1.
\end{equation}
Similarly, $\h_2/\|\h_2\|_2$ also lies in the set $\T'$, 
then by using Lemma \ref{L1-lemma1}, we have
\begin{equation}\label{th46-eq4}
\frac{1}{m}\|A\h_2\|_1\leq \left(\sqrt{\frac{2}{\pi}}+\frac{4\omega(\T)}{\sqrt{m}}+u\right)
   \|\h_2\|_2\leq \left(\sqrt{\frac{2}{\pi}}+\rho'+u\right)
   \|\h_2\|_2.
\end{equation}
Next, we estimate the term \textcolor{blue}{$\|\tilde{\e}\|_1=\|A\x^*-A\x^{\eta}+\e\|_1$ by \eqref{th46-eq4}}
\begin{equation}\label{ladeq23}
  \frac{1}{m}\|\tilde{\e}\|_1  \leq \frac{1}{m}\|A\h_2\|_1+\frac{1}{m}\|\e\|_1 
   \leq 
   \left(\sqrt{\frac{2}{\pi}}+\rho'+u\right)
   \|\h_2\|_2+\frac{1}{m}\|\e\|_1.
\end{equation}
Combining
\eqref{lad-eq2-1},
\eqref{lad-eq2-2},
\eqref{th46-eq4}
and 
\eqref{ladeq23}, 
we  deduce that
\begin{align}\label{ladeq60}
\left(\sqrt{\frac{2}{\pi}}-\rho'-u\right)\|\h_1\|_2 
\leq& \frac{1}{m}\|A\h_1\|_1 
\leq \frac{2}{m}\|\tilde{\e}\|_1\notag\notag\\
\leq&
2\left(\sqrt{\frac{2}{\pi}}+\rho'+u\right)\|\x^{\eta}-\x^*\|_2+\frac{2}{m}\|\e\|_1.
\end{align}
 Finally, we conclude that by \eqref{ladeq60}
 \textcolor{blue}{
\begin{align*}
 \|\h\|_2&\leq\|\h_1\|_2+\|\x^{\eta}-\x^*\|_2\\
 &\leq   \frac{2\left(\sqrt{\frac{2}{\pi}}+\rho'+u\right)}{\sqrt{\frac{2}{\pi}}-\rho'-u}\left(\frac{\eta}{f(\x^*)}-1\right)
\|\x^*\|_2+\frac{2}{\sqrt{\frac{2}{\pi}}-\rho'-u}\frac{\|\e\|_1}{m}+\left(\frac{\eta}{f(\x^*)}-1\right)
\|\x^*\|_2\\
&=\frac{3\sqrt{\frac{2}{\pi}}+\rho'+u}{\sqrt{\frac{2}{\pi}}-\rho'-u}\left(\frac{\eta}{f(\x^*)}-1\right)
\|\x^*\|_2+\frac{2}{\sqrt{\frac{2}{\pi}}-\rho'-u}\frac{\|\e\|_1}{m}.
\end{align*} 
 }
 
%
 
\vfill

\end{document}